\documentclass[a4paper]{article}
	
\usepackage[utf8]{inputenc}
\usepackage[T1]{fontenc}
\usepackage[english]{babel}
\usepackage{enumitem} \usepackage{comment} \usepackage{calrsfs} \DeclareMathAlphabet{\pazocal}{OMS}{zplm}{m}{n} 
\usepackage[numbers]{natbib} \usepackage{xspace} \usepackage{xargs} 
\usepackage{fullpage}

\usepackage[colorlinks=true]{hyperref}

\usepackage{mathtools} 
\usepackage{amssymb}
\usepackage{amsthm} \usepackage{bm} \usepackage{mathrsfs}  
\usepackage{caption}
\usepackage{float}
\usepackage{authblk}

\usepackage{xcolor}
\def\rev#1{\textcolor{black}{#1}}
\def\revbis#1{\textcolor{black}{#1}}

\renewcommand{\tilde}[1]{\widetilde{#1}}
\renewcommand{\bar}[1]{\overline{#1}}
\newcommand{\eps}{\varepsilon}
\newcommand{\RPuis}{\mathbf{R}\langle\eps\rangle}
\newcommand{\KK}{\mathbf{K}}
\newcommand{\KKbar}{\bar{\mathbf{K}}}
\newcommand{\CC}{\mathbf{C}}
\newcommand{\RR}{\mathbf{R}}
\newcommand{\QQ}{\mathbf{Q}}
\newcommand{\CCo}{\mathbb{C}}
\newcommand{\RRo}{\mathbb{R}}
\newcommand{\QQo}{\mathbb{Q}}

\renewcommand{\aa}{\bm{a}}
\newcommand{\bb}{\bm{b}}
\newcommand{\ff}{{\bm{f}}}
\renewcommand{\gg}{\bm{g}}
\newcommand{\hh}{{\bm{h}}}
\newcommand{\uu}{\bm{u}}

\newcommand{\ww}{\bm{w}}
\newcommand{\xx}{{\bm{x}}}
\newcommand{\yy}{{\bm{y}}}
\newcommand{\zz}{{\bm{z}}}

\newcommand{\FF}{{\bm{F}}}

\newcommand{\Wo}{W^\circ}

\newcommandx{\rinf}[2][1 = \bphi]{_{\mid #1 < #2}}
\newcommandx{\rinfeq}[2][1 = \bphi]{_{\mid #1 \leq #2}}
\newcommandx{\req}[2][1 = \bphi]{_{\mid #1 = #2}}
\newcommandx{\rinto}[2][1 = \bphi]{_{\mid #1 \in #2}}

\renewcommand{\phi}{\varphi}

\newcommand{\Rcal}{\mathcal{R}}
\newcommand{\Pcal}{\mathcal{P}}
\newcommand{\Dcal}{\mathcal{D}}

\newcommand{\Ppaz}{\pazocal{P}}
\newcommand{\Spaz}{\pazocal{S}}

\newcommand{\Opaz}{\pazocal{O}}
\newcommand{\Bpaz}{\pazocal{B}}
\newcommand{\Bfrak}{\mathfrak{B}}
\newcommand{\et}{\quad \text{and} \quad}
\newcommand{\V}{\bm{V}}
\newcommand{\I}{\bm{I}}
\newcommandx{\HH}[2][1= , 2= ]{$\mathbf{H}_{\bm{#1}}$#2\xspace}
\newcommandx{\HHp}[2][1= , 2= ]{$\mathbf{H'_{#1}}$#2\xspace}
\newcommand{\RM}{$\mathrm{RM}$\xspace}
\newcommandx{\RMpu}[1][1 = u]{$\mathrm{RM}(#1)$\xspace}
\newcommand{\SW}{\rev{S_i}}

\def\sfA{\mathsf{A}}
\def\sfB{\mathsf{B}}
\def\sfC{\mathsf{C}}
\def\sfP{\mathsf{P}}
\def\bphi{\bm{\phi}}
\newcommand{\balpha}{\bm{\alpha}}
\newcommandx{\map}[3][3=]{\bm{#1}_{#2}^{#3}}
\newcommand{\phiun}{\map{\phi}{1}}

\newcommand{\SA}{semi-\allowbreak algebraic\xspace}
\newcommand{\SAC}{semi-\allowbreak algebraically connected\xspace}
\newcommand{\SACC}{semi-\allowbreak algebraically connected component\xspace}
\newcommand{\SACCs}{semi-\allowbreak algebraically connected components\xspace}

\DeclareMathOperator{\ext}{ext}

\DeclareMathOperator{\sing}{sing}
\DeclareMathOperator{\rank}{rank}
\DeclareMathOperator{\jac}{Jac}
\DeclareMathOperator{\reg}{reg}
\DeclareMathOperator{\id}{id}

\DeclareMathOperator{\limeps}{lim_\eps}

\theoremstyle{plain}
\newtheorem{theorem}{Theorem}[section]
\newtheorem{lemma}[theorem]{Lemma}
\newtheorem{proposition}[theorem]{Proposition}
\newtheorem{corollary}[theorem]{Corollary}

\theoremstyle{definition}
\newtheorem{definition}[theorem]{Definition}
\newtheorem{example}{Example}
\newtheorem*{notation}{Notation}

\theoremstyle{remark}
\newtheorem*{remark}{Remark}

\begin{document}
\author[1]{Rémi \textsc{Pr\'ebet}}
\author[1]{Mohab \textsc{Safey El Din}}
\author[2]{\'Eric \textsc{Schost}}
\affil[1]{Sorbonne Universit\'e, LIP6 CNRS UMR 7606, Paris, France}
\affil[2]{University of Waterloo, David Cheriton School of Computer Science,
  Waterloo ON, Canada}
\date{\today}
\title{Computing roadmaps in unbounded smooth real algebraic sets I: 
connectivity results}

\maketitle 
 
\begin{abstract}
Answering connectivity queries in real algebraic sets is a fundamental problem 
in effective real algebraic geometry that finds many applications in e.g. 
robotics where motion planning issues are topical.
This computational problem is tackled through the computation of so-called 
\emph{roadmaps} which are real algebraic subsets of the set $V$ under study, of 
dimension at most one, and which have a connected intersection with all 
\SACCs of $V$.
Algorithms for computing roadmaps rely on statements establishing 
connectivity properties of some well-chosen subsets of $V$, assuming that $V$ 
is bounded.

In this paper, we extend such connectivity statements by dropping the 
boundedness assumption on $V$. This exploits properties of so-called 
\emph{generalized polar varieties}, which are critical loci of $V$ for some 
well-chosen polynomial maps.
\end{abstract}

\maketitle 

\section{Introduction}
Let $\QQ$ be a real field of real closure $\RR$ and let $\CC$ be its algebraic
closure (one can think about $\QQo$, $\RRo$ and $\CCo$ instead, for the sake of
understanding) and let $n \geq 0$ be an integer. An algebraic set $V \subset
\CC^n$ defined over $\QQ$ is the solution set in $\CC^n$ to a system of
polynomial equations \rev{in $n$ variables} with coefficients in $\QQ$. A real
algebraic set defined over $\QQ$ is the set of solutions in $\RR^n$ to a system
of polynomial equations \rev{in $n$ variables} with coefficients in $\QQ$. It is
also the real trace $V\cap \RR^n$ of an algebraic set $V\subset\CC^n$. Real
algebraic sets have finitely many connected components \cite[Theorem
2.4.4.]{BCR2013}. Counting these connected components \cite{GV1992,VG1990} or
answering connectivity queries over $V\cap \RR^n$ \cite{SS1983} finds many
applications in e.g. robotics \cite{Ca1988, CSS2020, Wenger07, NaSc17, 
CPSSW2022}.

Following \cite{Ca1988,Ca1993}, such computational issues are tackled
by computing a real algebraic subset of $V\cap \RR^n$, defined over
$\QQ$, which has dimension at most one and a connected intersection
with all connected components of $V$ and contains the input query
points. In~\cite{Ca1988}, Canny called such a subset a \emph{roadmap}
of $V$.

The effective construction of roadmaps, given a defining system for $V$, relies
on connectivity statements which allow one to define real algebraic subsets of
$V\cap \RR^n$, of smaller dimension than \rev{that} of $V\cap \RR^n$, and that 
have 
a connected intersection with the connected components of $V\cap \RR^n$. Such
existing statements in the literature make the assumption that $V$ has finitely
many singular points and $V\cap \RR^n$ is bounded. In this paper, we focus on
the problem of obtaining similar statements by dropping the boundedness
assumption. 
We prove a new connectivity statement which generalizes the one of 
\cite{SS2017} 
to the unbounded case and will be used in a separate paper to obtain 
asymptotically faster algorithms for computing roadmaps. 
We start by recalling the state-of-the-art connectivity statement,
which allows us to introduce some material we need to state our main result.

\paragraph*{State-of-the-art overview}
We start by introducing some terminology.
Recall that an \emph{algebraic set} $V\subset\CC^n$ is the set of solutions of 
a 
finite system of polynomials equations.
It can be uniquely decomposed into finitely many \emph{irreducible components}. 
When all these components have the same dimension $d$, we say that $V$ is 
\emph{$d$-equidimensional}.
Those points $\yy\in V$ at which the Jacobian matrix of a finite set of 
generators
of its associated ideal has rank $n-d$ are called \emph{regular} points and 
the 
set of those points is denoted by $\reg(V)$. 
The others are called \emph{singular} points; the set of singular points of $V$
(its singular locus) is denoted by $\sing(V)$ and is an algebraic subset of $V$.
We refer to \cite{Sh1994} for definitions and propositions about algebraic sets.

A \emph{\SA set} $S\subset\RR^n$ is the set of solutions of a finite system of 
polynomial equations and inequalities. We say that $S$ is \emph{\SAC} if for 
any $\yy, \yy' \in S$, $\yy$ and $\yy'$ can be connected by a \emph{\SA path} 
in $S$, that is a continuous \SA function $\gamma\colon[0,1]\to S$ such that 
$\gamma(0)=\yy$ and $\gamma(1)=\yy'$.
A \SA set $S$ can be decomposed into finitely many \emph{\SACCs} which are \SAC 
\SA sets that are both closed and open in $S$.
Finally, for a \SA set $S\subset \RR^n$, we denote by $\bar{S}$ its closure for 
the Euclidean topology on $\RR^n$.
We refer to \cite{BPR2006} and \cite{BCR2013} for definitions and propositions 
about \SA sets and functions.

Let $0\leq d \leq n$ and $V\subset \CC^n$ be a $d$-equidimensional algebraic set
such that $\sing(V)$ is finite. For $1 \leq i \leq n$, let $\pi_i$ be the
canonical projection: 
\[
 \pi_i\colon (\yy_1,\dotsc,\yy_n) \: \longmapsto \: (\yy_1,\dotsc,\yy_i)
\]

For a polynomial map $\bphi\colon\CC^n \rightarrow \CC^m$ a point $\yy\in V$ is
a \emph{critical point} of $\bphi$ if $\yy \in \reg(V)$ and \rev{the
  differential of the restriction of $\phi$ to $V$ at $\yy$}, denoted by $d_\yy
\bphi$, is not \rev{surjective}, that is
\[ 
 d_\yy \bphi(T_\yy V) \subsetneq \CC^m,
\]
where \rev{$T_\yy V$ denoted the tangent space to $V$ at $\yy$}. We will denote
by $\Wo(\bphi,V)$ the set of the critical points of $\bphi$ on $V$. A
\emph{critical value} is the image of a critical point. We \revbis{put} 
$K(\bphi,V) = \Wo(\bphi,V) \cup \sing(V)$. The points of $K(\bphi,V)$ are 
called the \emph{singular points} of $\bphi$ on $V$. Figure~\ref{fig:critlocus} 
show
examples of such critical loci.

\begin{figure}[h]\centering
\begin{minipage}[c]{0.5\linewidth}
  \includegraphics[width=\linewidth]{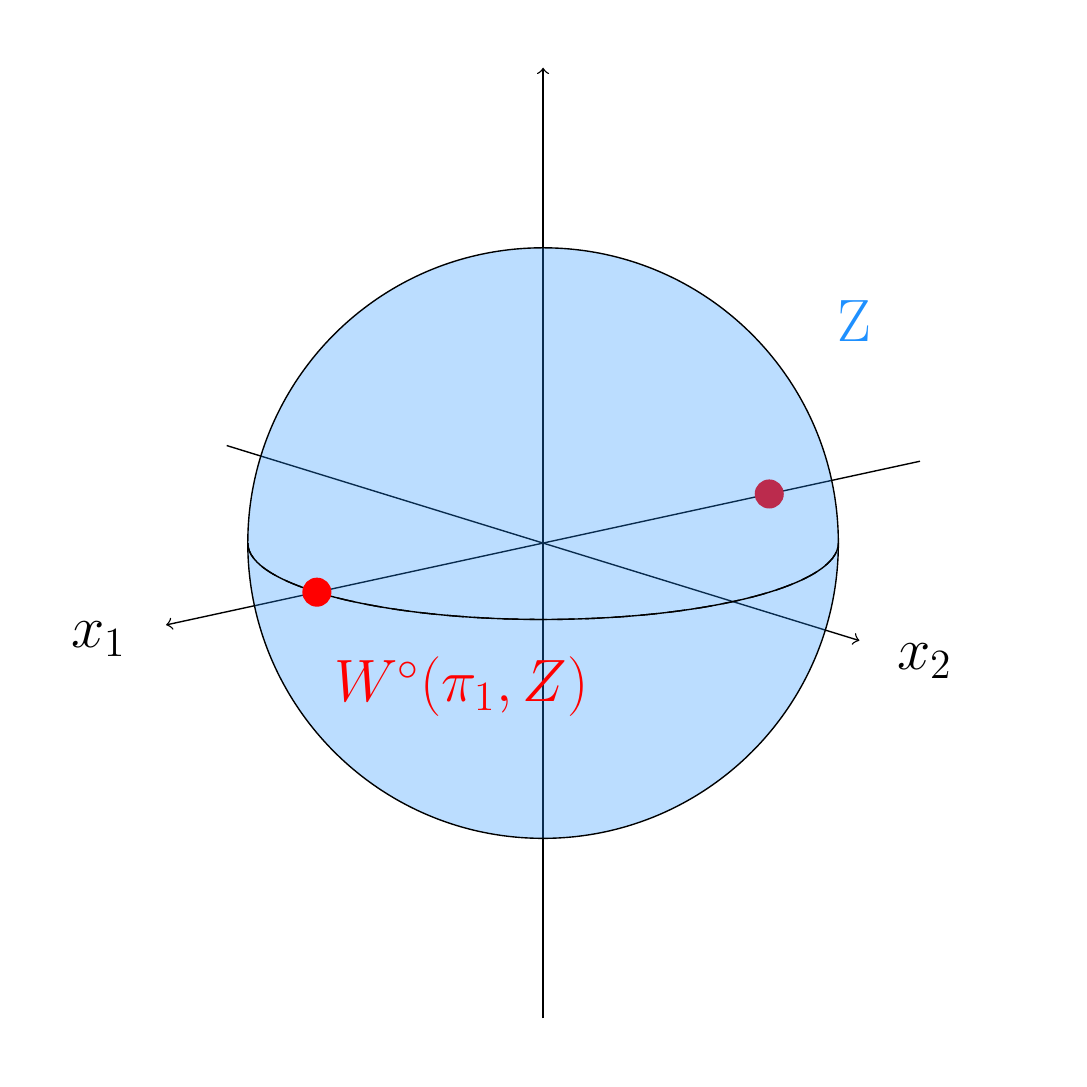}
\end{minipage}\hfill
\begin{minipage}[c]{0.5\linewidth}
 \includegraphics[width=\linewidth]{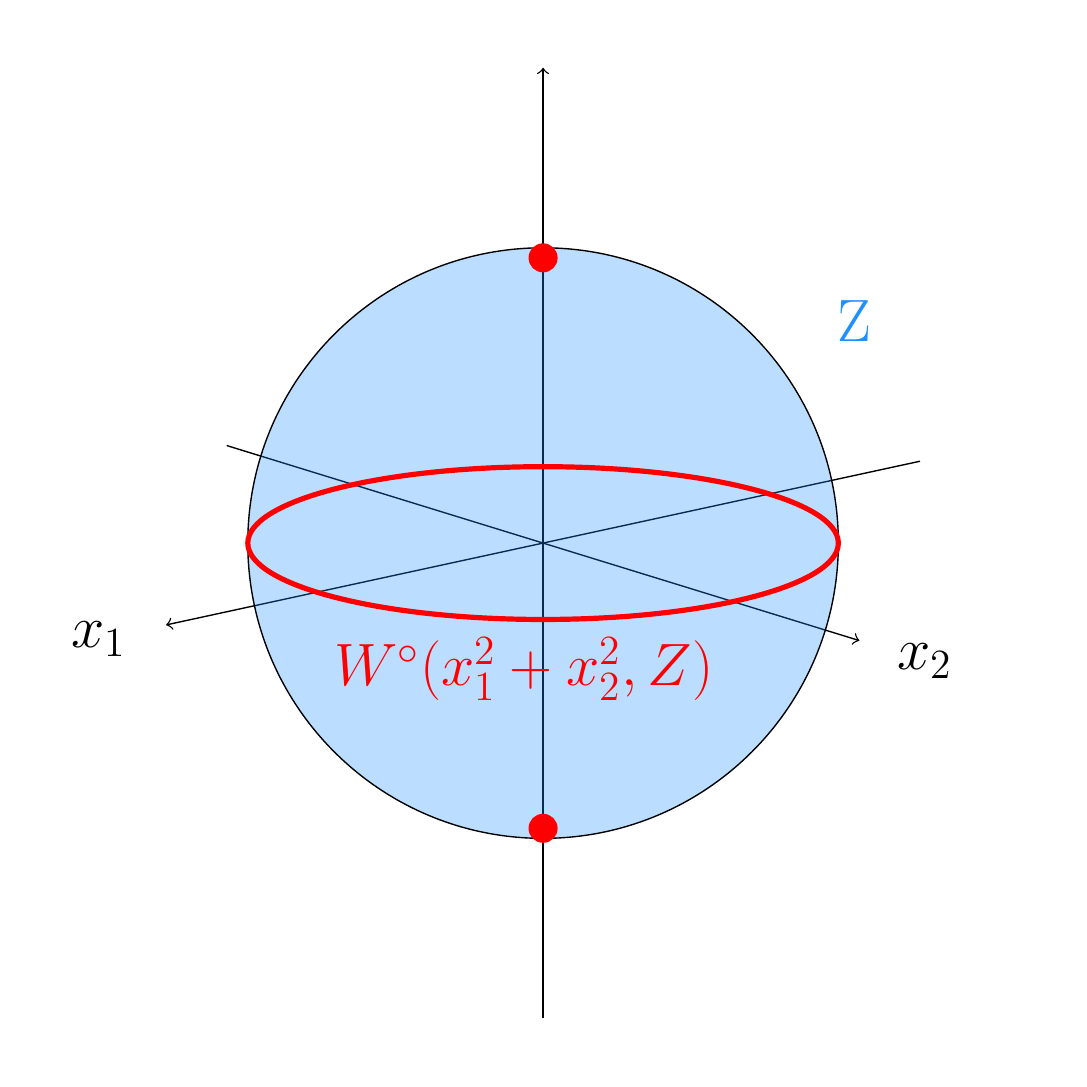}
\end{minipage}
\caption{Real trace of the critical locus on a sphere $Z$ for: the 
projection on the first coordinate $\pi_1$ (left); 
the polynomial map $\phi$ associated to $x_1^2+x_2^2 \in \RRo[x_1,x_2,x_3]$ 
(right).
\revbis{Let $\xx = (\xx_1,\xx_2,\xx_3) \in Z$.
The differential of the restriction of $\pi_1$ to $Z$ at $\xx$ is 
the restriction of $\pi_1$ to $T_{\xx}Z$.
The image is not $\CC$ if, and only if, $T_{\xx}Z$ is orthogonal to the 
$x_1$-axis, so that critical points of the restriction of $\pi$ to $Z$ occur at 
$(\pm1,0,0)$.
Besides, the differential of the restriction of $\phi$ to $Z$ at $\xx$ 
is the restriction of $-2x_3\cdot\pi_3$ to $T_{\xx}Z$.
Hence, $\xx$ is a critical point of the restriction of $\phi$ to $Z$ if, and 
only if, either $\xx_3=0$ or $T_{\xx}Z$ is orthogonal the $x_3$-axis.}}
\label{fig:critlocus}
\end{figure}

For $1\leq i \leq d$ we denote by $W(\pi_i,V)$ the $i$-th \emph{polar variety} 
defined as the Zariski closure of the critical locus $\Wo(\pi_i,V)$ 
of the restriction of $\pi_i$ to $V$.
Further, we extend this definition by considering $\bphi = (\phi_1, \ldots,
\phi_n) \subset \QQ[x_1,\dotsc,x_n]$ and, for $1\leq i \leq n$, the map
  \begin{equation}\label{eqn:defphii}
    \begin{array}{cccc}
        \map{\phi}{i}\colon &\CC^n& \longrightarrow &\CC^i\\
        &\yy & \mapsto & (\phi_1(\yy), \ldots, \phi_i(\yy))
    \end{array}
  \end{equation}
Following the ideas of \cite{BGHP2004,BGHP2005, BGHSS2010} we denote similarly 
$W(\map{\phi}{i},V)$ the $i$-th \emph{generalized polar variety} defined as the 
Zariski closure of the critical locus $\Wo(\map{\phi}{i},V)$ of the restriction 
of $\map{\phi}{i}$ to $V$. 
We recall below \cite[Theorem 14]{SS2011} (see also \cite[Proposition 
3.3]{BRSS2014} for a slight variant of it), making use of polar varieties to 
establish connectivity statements.

\smallskip
\emph{
For $2\leq i \leq d$, assume that the following holds:
\begin{itemize}
 \item $V\cap \RR^n$ is bounded;
 \item $W(\pi_i, V)$ is either empty or $(i-1)$-equidimensional and smooth
   outside $\sing(V)$;
 \item $W(\pi_1, W(\pi_i, V))$ is finite;
 \item for any $\yy\in \CC^{i-1}$, $\pi_{i-1}^{-1}(\yy)\cap V$ is either empty 
or
   $(d-i+1)$-equidimensional.
\end{itemize}
Let
\[
K_i = W(\pi_1,W(\pi_i, V)) \cup \sing(V) 
\et
F_i = \pi_{i-1}^{-1}(\pi_{i-1}(K_i)) \cap V.
\]
Then,\,the\,real\,trace\,of $W(\pi_i, V)\cup F_i$ has\,a\,non-empty\,and \SAC 
intersection\,with\,each\,\SACC\,of $V\cap\RR^n$.
}

For the special case $i = 2$, this result was  originally proved by Canny in
\cite{Ca1988, Ca1988bis}. A variant of it, again assuming $i=2$, is given for
general \SA sets in \cite{Ca1993,Ca1991}. By dropping the restriction $i=2$, the
result in \cite[Theorem 14]{SS2011} allows one more freedom in the choice of
$i$, and then, in the design of roadmap algorithms to obtain a better 
complexity.
The rationale is as follows.

Restricting to $i=2$, one expects (up to some linear change of variables or
other technical manipulations) a situation where $W(\pi_2, V)$ has dimension at
most $1$ and $F_2$ has dimension $d-1$ (see e.g. \cite[Lemma 31]{SS2011}). To
obtain a roadmap for $V\cap \RR^n$ one is led to call recursively roadmap
algorithms with input systems defining the fibers $F_i$'s. Hence, the depth of 
the
recursion is $n$. Besides, letting $D$ be the maximum degree of input equations
defining $V$, roughly speaking each recursive call requires $(nD)^{O(n)}$
arithmetic operations in $\QQ$ while the size of the input data grows by
$(nD)^{O(n)}$ according to \cite[Proposition 33]{SS2011}. Consequently, one
obtains roadmap algorithms using $(nD)^{O(n^2)}$ arithmetic operations in $\QQ$.

In \cite{SS2011}, using a baby steps/giant steps strategy, it is showed 
that one can take $i\simeq \sqrt{d}$ and then have a depth of the recursion 
$\simeq \sqrt{d}$. It is also proved that each recursive step needed to compute 
systems encoding $K_i$ and $F_i$ requires at most $(nD)^{O(n)}$ arithmetic 
operations in $\QQ$, while the size of the input data grows by $(nD)^{O(n)}$. 
All in all, up to technical details that we skip, one obtains roadmap 
algorithms using $(nD)^{O(n\sqrt{n})}$ 
arithmetic operations in $\QQ$. 
Finally, in \cite{SS2017}, it is shown how to apply \cite[Theorem 14]{SS2011}
with $i\simeq \frac{d}{2}$ so that the depth of the recursion becomes $\simeq
\log_2(d)$. Hence, proceeding as in \cite{SS2011}, an algorithm using
$(nD)^{6n\log_2(d)}$ arithmetic operations in $\QQ$ is obtained
in~\cite{SS2017}.

\rev{Such connectivity results and the algorithms that derive from them  are at 
the foundation
  of
  many
  implementations for answering connectivity queries in real
  algebraic sets. As far as we know, the first one was reported in 
\cite{MeSa06},
  showing that, at that time, basic computer algebra tools were mature enough to
  implement rather easily roadmap algorithms. More recently, practical
  results were reported applications of roadmap algorithms to kinematic
  singularity analysis in \cite{CSS2020, CSS23}, showing the interest of
  developing roadmap algorithms beyond applications to motion planning. In 
parallel,
 the interest in roadmap algorithms keeps growing as they have 
  also been adapted to the numerical side \cite{ReCh14, CWF19}. This 
illustrates the
  interest of improving roadmap algorithms and the connectivity results they
  rely on.}

Dropping the boundedness assumption in this  scheme was done in
\cite{BR2014, BRSS2014} using infinitesimal deformation techniques. The
algorithms proposed use respectively $(nD)^{O(n\sqrt{n})}$ and
$(nD)^{O(n\log^2(n))}$ arithmetic operations in $\QQ$. This induces a growth of
intermediate data; the algorithm is not polynomial in its output
size, \rev{which is $(nD)^{O(n\log(n))}$}.

\rev{In non-compact \revbis{cases}, one could also study the intersection of
  $V$ with either $[-c, c]^{n}$ or a ball of radius $c$, for $c$ large
  enough, but we would then have to deal with semi-algebraic sets instead
  of real algebraic sets, in which case \cite[Theorem 14]{SS2011} is
  still not sufficient.}

\rev{In order to ultimately obtain an algorithm dealing with unbounded smooth
  real algebraic sets with a complexity similar to that of
  \cite{SS2017}, the goal of this paper is instead to provide a new
  connectivity statement with no boundedness assumption and the same
  freedom brought by the one of \cite{SS2011}.}

\paragraph*{Main result}\label{ssec:contrib}
Let $V \subset \CC^n$ be an
algebraic set defined over $\QQ$ and $d>0$ be an integer. We say that $V$
satisfies assumption $(\sfA)$ when
\begin{description}
\item[$(\sfA)$] $V$ is $d$-equidimensional and its singular locus $\sing(V)$ is
  finite.
\end{description}

Recall that we say that a map $\psi\colon Y\subset\RR^n\to Z\subset\RR^m$ is a 
proper map if, for every closed (for Euclidean topology) and bounded subset 
$Z'\subset Z$, $\psi^{-1}(Z')$ is a closed and bounded subset of $Y$.
\rev{For $\bphi = (\phi_1, \ldots, \phi_n) \subset \QQ[x_1,\dotsc,x_n]$, and
with $\map{\phi}{i}$ the induced map defined in \eqref{eqn:defphii}, for 
$1\leq i \leq n$, we say that $\bphi$ satisfies assumption $(\sfP)$ when}
\begin{description}
 \item[$(\sfP)$] the restriction of the map $\map{\phi}{1}$ to $V\cap \RR^n$ is 
proper and bounded from below. 
\end{description}

We denote by $W_i = W(\map{\phi}{i},V)$ the Zariski closure of the set of
critical points of the restriction of $\map{\phi}{i}$ to $V$.
For $2\leq i \leq d$ \rev{and $\bphi$ as above}, we say that $(\bphi, i)$ 
satisfies assumption $(\sfB)$ when 
  \begin{enumerate}[label=$(\sfB_\arabic*)$] 
  \item $W_i$ is either empty or
    $(i-1)$-equidimensional and smooth outside $\sing(V)$;
  \item for any $\yy = (\yy_1, \ldots, \yy_i)\in \CC^i$,
    $V\cap\map{\phi}{i-1}[-1](\yy)$ is either empty or
    $(d-i+1)$-equidimensional. 
\end{enumerate}

Note that when $\sfB_{1}$ holds, $\sing(W_i)$ and critical loci of polynomial
maps restricted to $W_i$ are well-defined. 
\rev{For $\SW$  a finite subset of $V$,} we say 
that $\SW$ satisfies assumption $(\sfC)$ when
\begin{enumerate}[label=$(\sfC_\arabic*)$]
 \item $\SW$ is finite;
 \item $\SW$ \rev{has a non-empty intersection with} every \SACC of 
$W(\phiun,W_i)\cap\RR^n$.
\end{enumerate}
Finally, using a construction similar to the one used in \cite[Theorem
14]{SS2011}, we let 
\[
 K_i = W(\map{\phi}{1}, V) \cup \SW \cup \sing(V)
 \quad \et \quad F_i = \map{\phi}{i-1}[-1](\map{\phi}{i-1}(K_i))\cap V.
\]

\begin{theorem}\label{thm:mainresult}
  \indent For $V,d$, $i$ in $\{1,\dots,d\}$, $\bphi$ and $S_i$ as
  above, and under assumptions $(\sfA)$, $(\sfB)$, $(\sfC)$ and
  $(\sfP)$, the subset $W_i\cup F_i$ has a non-empty and \SAC
  intersection with each \SACC of $V\cap\RR^n$.
\end{theorem}

\rev{The proof structure of the above result follows a pattern similar to the
  one of~\cite{SS2011}. Its foundations rely on the following basic idea, 
 sweeping the ambient space with level sets of $\map{\phi}{1}$, having a look
  at the connectivity of $V\cap {\map{\phi}{1}}^{-1}(]-\infty, a])$ and $\left(
    W_i\cup F_i \right)\cap {\map{\phi}{1}}^{-1}(]-\infty, a])$. The bulk of the
  proof consists in showing that these connectivities are the same. When one
  does not assume that $i=2$ but \revbis{does assume boundedness}, one can take
  for $\map{\phi}{1}$ a linear projection, so that its level sets are
  hyperplanes. In this context, the proof in~\cite{SS2011} also introduces 
  ingredients such as Thom's isotopy lemma, which can be used thanks to the
  boundedness assumption. Dropping the boundedness assumption makes these steps
  more difficult and requires us to use a quadratic form for $\map{\phi}{1}$
  to ensure a properness property. This in turn makes the
  geometric analysis more involved since now, the level sets of $\map{\phi}{1}$
  are not hyperplanes anymore.}

\paragraph*{\rev{Structure of the paper}} \rev{Section~\ref{sec:preliminaries} 
provides the necessary background on
  algebraic sets and polar varieties needed to follow the proof of
  Theorem~\ref{thm:mainresult}. Section~\ref{sec:auxiliary} proves two
  auxiliary results which analyze the connectivity of fibers of some polynomial
  maps. These are used in the proof of Theorem~\ref{thm:mainresult}, which is
  given in Section~\ref{sec:proof}. Finally, in Section~\ref{sec:perspectives},
  we sketch how Theorem~\ref{thm:mainresult} will be used to design new roadmap
  algorithms in upcoming work.}

 \section{Preliminaries}\label{sec:preliminaries}
\paragraph*{Basic properties of algebraic sets}

Recall that given a finite set of polynomials $\gg \subset \CC[x_1,\dotsc,x_n]$ 
we denote by $\V(\gg) \subset \CC^n$ the algebraic set defined as the vanishing 
locus of $\gg$.
For $\yy \in \CC^n$, we denote by $\jac_\yy(\gg)$ the Jacobian matrix of 
$\gg$ evaluated at $\yy$. 
Conversely, given an algebraic set $V\subset \CC^n$, we denote by $\I(V)$ the 
\emph{ideal of $V$}, that is the ideal of $\CC[x_1,\dotsc,x_n]$ of polynomials 
vanishing on $V$. 
Such an ideal is finitely generated by the Hilbert basis theorem.

Let $X \subset \CC^n$ and $Y \subset \CC^m$ be algebraic sets and $\KK\subset 
\CC$ be a subfield. 
A map $\alpha \colon X \rightarrow Y$ is a \emph{regular} map defined over 
$\KK$ if there exists $(f_1,\dotsc,f_m) \subset \KK[x_1,\dotsc,x_n]$ such that 
$\alpha(\yy) = (f_1(\yy), \dotsc, f_m(\yy))$ for all $\yy \in X$.
A regular map $\alpha\colon X\to Y$ is an \emph{isomorphism} defined over $\KK$ 
if there exists a regular map $\beta\colon Y \to X$, defined over $\KK$, such 
that $\alpha \circ \beta = \id_Y$ and $\beta \circ \alpha = \id_X$, where 
$\id_Z\colon Z \rightarrow Z$ is the identity map on $Z$.
We refer to \cite{Sh1994} for further details on these notions.
The following result is straightforward.
\begin{lemma}\label{lem:isoconnected}
Let $Y \subset \CC^n$ and $Z \subset \CC^m$ be two algebraic sets.
Let $\alpha\colon Y \rightarrow  Z$ be an isomorphism of algebraic sets defined 
over $\RR$. Then the \SAC subsets of $Y\cap \RR^n$ and $Z\cap \RR^m$ are in 
correspondence through $\alpha$.
\end{lemma}

\paragraph*{Critical points of a polynomial map}

The following lemma from \cite[Lemma A.2]{SS2017} provides an algebraic 
characterization of critical points.
\begin{lemma}[Rank characterization]\label{lem:caraccritrank}
 Let $Z \subset \CC^n$ be a $d$-equidimensional algebraic set and 
$\gg=(g_1,\dotsc,g_p)$ be generators of $\I(Z)$.
Let $\bphi\colon Z \rightarrow \CC^m$ be a polynomial map, then the following 
holds.
\begin{align*}
 \Wo(\bphi, Z) &= \left\lbrace \yy \in Z \mid
 \begin{array}{l}
    \rank(\jac_\yy(\gg)) = n-d \\
    \text{and} \quad \rank(\jac_\yy([\gg,\bphi])) < n-d+m
 \end{array}\right\rbrace;\\
 K(\bphi, Z) &= \lbrace \yy \in Z \mid \rank(\jac_\yy([\gg,\bphi])) < 
 n-d+m \rbrace.
\end{align*}
\end{lemma}

Let us present a direct consequence of this result, which gives a more 
effective criterion for the singular points of a polynomial map.
Let $\bphi = (\phi_1, \ldots, \phi_n) \subset \CC[x_1,\dotsc,x_n]$ and
$\map{\phi}{i}$ be the deduced map defined as in \eqref{eqn:defphii} for 
$1\leq i \leq n$. 

\begin{lemma}\label{lem:caraccritminor}
Let $Z \subset \CC^n$ be a $d$-equidimensional variety and 
$\gg$ be a finite set of generators of $\I(Z)$. 

Then for $1\leq i\leq n$, $K(\bphi_i, Z)$ is the algebraic subset
of $Z$ defined by the vanishing of $\gg$ and the $(p+i)$-minors of 
$\jac([\gg,\bphi_i])$, where $p=n-d$.
\end{lemma}
\begin{proof}
One directly deduces from Lemma~\ref{lem:caraccritrank}
that $K(\bphi_i, Z)$ is exactly the intersection of $Z$, the zero-set of $\gg$, 
with the set of points $\yy\in\CC^n$ where $\rank(\jac_\yy([\gg,\bphi_i]))<p+i$.
The latter set is the zero-set of the 
$(p+i)$-minors of $\jac([\gg,\bphi_i])$.
\end{proof}

\begin{definition}[Polar variety]\label{def:polarvariety}
  Let $Z\subset \CC^n$ be a $d$-equidimensional algebraic set, and let $1\leq i
  \leq n$. \rev{As above, let $\bphi = (\phi_1, \ldots, \phi_n) \subset
    \CC[x_1,\dotsc,x_n]$ and $\map{\phi}{i}$ be the induced map, defined by
    $(\phi_1, \ldots, \phi_i)$.} We denote by $W(\bphi_i,Z)$
  the Zariski closure of $\Wo(\bphi_i,Z)$. It is called a \emph{generalized
    polar variety of $Z$}. Remark that
\[ 
 \Wo(\bphi_i,Z) \:\subset\: W(\bphi_i,Z) \:\subset\: K(\bphi_i,Z) \:\subset\: Z
\] 
by minimality of the Zariski closure. 
Hence $K(\bphi_i,Z) = W(\bphi_i,Z) \cup \sing(Z)$ but the union is not 
necessarily disjoint.
\end{definition}

 \section{Connectivity and critical values}\label{sec:auxiliary}
In this section we consider for $n \geq 1$ an equidimensional algebraic set $Z
\subset \CC^n$ of dimension $d>0$. We are going to prove two main connectivity
results on the \SACCs of \rev{$Z\cap \RR^{n}$} through some polynomial map.
These results, along with \rev{ingredients} of Morse theory \rev{such as
  critical loci and critical values of polynomial maps}, will be essential in
the proof of Theorem~\ref{thm:mainresult}. Most of the results presented here
are generalizations of those given in \cite[Section 3]{SS2011} in the
unbounded case, replacing projections by suitable polynomial maps.

\subsection{Connectivity changes at critical values}
The main result of this \rev{subsection} is to prove the following proposition,
which deals with the connectivity changes of \SACCs in the neighbourhood of
singular values of a polynomial map.

\rev{Let $X$ be a subset of $\CC^n$, $U \subset \RR$ and $f \in 
\RR[x_1,\dotsc,x_n]$.
With a slight abuse of notation, we still denote by $f$ the polynomial map 
$\yy\in\CC^n\mapsto f(\yy)\in \CC$, and we write $X\rinto[f]{U} = X \cap 
f^{-1}(U) \cap \RR^n$.
In particular if $u \in \RR$ we note
\[
    X\rinf[f]{u} = X\rinto[f]{]-\infty,u[},
    \quad X\rinfeq[f]{u} = X\rinto[f]{]-\infty,u]} 
    \et X\req[f]{u}=X\rinto[f]{\{u\}}.
\]}\vspace*{-1em}
\begin{proposition}\label{prop:firstresult}
  Let $\bphi\colon \CC^n \rightarrow \CC$ \rev{be} a regular map defined over
  $\RR$. Let $A \subset \RR^k$ be a \SAC semi-algebraic set, and $u \in \RR$ and
\[
  \gamma\colon A \rightarrow Z\rinfeq{u} - \left(Z\req{u} \cap K(\bphi,Z)\right)
\]
be a continuous semi-algebraic map. 
Then there exists a unique \SACC $B$ of $Z\rinf{u}$ such that $\gamma(A) 
\subset \overline{B}$.
\end{proposition}

\rev{Let us start by recalling a definition from \cite[Section 3.5]{BPR2006}.
  Let $U \subset \RR^k$ a semi-algebraic open set and $V \subset \RR^l$ a
  semi-algebraic set. The set of semi-algebraic functions from $U$ to $V$
  which admit partial derivatives up to order $m\geq 0$ is denoted by
  $\Spaz^m(U,V)$. The set $\Spaz ^\infty(U,V)$ is the intersection of all the
  sets $\Spaz^m(U,V)$ for $m \geq 0$. The ring $\Spaz^\infty(U,\RR)$ is called
  the ring of \emph{Nash functions}. }

\begin{notation}
  In this \rev{subsection} we fix a regular (polynomial) map
  $\bphi\colon \CC^n \rightarrow \CC$ defined over $\RR$. With a
  slight abuse of notation, the underlying polynomial in
  $\RR[x_1,\dotsc,x_n]$ will be denoted in the same manner.
\end{notation}

We start by proving an extended version of \cite[Lemma
  6]{SS2011}. This can be seen as the founding stone of all the
connectivity results presented in this paper. \rev{
  For any $\yy\in Z\cap\RR^n-K(\bphi,Z)$, it shows the
  existence of a regular map $\balpha : Z \to \CC^{n+1}$ such that
  $Z$ and $\balpha(Z)$ are isomorphic, with $\pi_1 \circ \balpha =
  \bphi$ on $\balpha(Z)$ and that 
  there is an open Euclidean neighborhood $N$ of $\balpha(\yy)$ such
  that the implicit function theorem applies to $\balpha(Z)\cap N$.}
(Recall that an open Euclidean neighborhood of a point $\yy \in \RR^n$
is any subset of $\RR^n$ that contains $\yy$ and is open for the
Euclidean topology on $\RR^n$.)

\begin{lemma}\label{lem:implicitparam}
  Let $\yy = (\yy_1,\dotsc,\yy_n)$ be in $Z \cap \RR^n - K(\bphi, Z)$. Then,
  there exists a \rev{regular} map $\balpha\colon Z \rightarrow \CC^{n+1}$
  such that the following holds :
\begin{enumerate}[label=\alph*)] 
\item there exist open Euclidean neighborhoods $N'\subset \RR^d$ of 
$\pi_d(\alpha(\yy))$ and $N \subset \RR^{n+1}$ of $\balpha(\yy)$, and  a 
continuous semi-algebraic map $\ff \colon N' \to \RR^{n+1-d}$
such that:
\[
 \balpha(Z) \cap N = \big\{ (\zz',\ff(\zz')) \mid \zz' \in N' \big\};
\]
 \item $\balpha\colon Z\to\balpha(Z)$ is an isomorphism of algebraic sets 
defined over $\RR$;
 \item $\bphi \circ \balpha^{-1} = \pi_1$ on $\balpha(Z)$.
\end{enumerate}
\end{lemma}
\begin{proof}
  \rev{Let $\Opaz_{\yy} \subset \RR^n$ be an open Euclidean neighborhood of 
$\yy$
and let $\gg = (g_1, \dotsc, g_{n-d})$ be an $(n-d)$-tuple of polynomials in
    $\CC[x_1,\dotsc,x_n]$, such that $Z \cap \Opaz_{\yy} = \V(\gg) \cap
    \Opaz_{\yy}$ and $\jac_\yy(\gg)$ has full rank $n-d$. Such a $\Opaz_{\yy}$
    and $\gg$ are given by \cite[Proposition 3.3.10]{BCR2013} since $\yy$ is in
    $\reg(Z)$. Also, since $\yy \notin W(\bphi,Z)$,} there exists a non-zero
  $(n-d+1)$-minor of $\jac_\yy([\gg,\bphi])$ by Lemma~\ref{lem:caraccritminor}.
  Therefore, there exists a permutation $\sigma$ of $\{1,\dotsc,n\}$ such that
  the matrix
\[
    \left[
    \begin{array}{c}
    \frac{\partial \gg}{\partial x_{\sigma(i)}}(\yy)\\
    \frac{\partial \bphi}{\partial x_{\sigma(i)}}(\yy) 
    \end{array}
    \right]_{d \leq j \leq n}
\]
is invertible. Let $x_0$ be a new variable and define $\hh$ as the following 
finite subset of polynomials of $\RR[x_0,x_1,\dotsc,x_n]$,
\[
 \hh = (\tilde{\gg},\tilde{\bphi})= 
\left( \gg(\sigma^{-1}\cdot(x_1,\dotsc,x_n)), 
\bphi(\sigma^{-1}\cdot(x_1,\dotsc,x_n))-x_0 \right)
\]
where $\tau \cdot (x_1,\dotsc,x_n) = (x_{\tau(1)},\dotsc,x_{\tau(n)})$ for any 
permutation $\tau$ of $\{1,\dotsc,n\}$. 
Hence,
\[
  \V(\hh)\cap (\RR \times \Opaz_{\yy}) = \big\{
  (\bphi(\rev{\zz}),\sigma\cdot\rev{\zz}) \mid \rev{\zz} \in Z \cap
  \Opaz_{\yy}\big\} \subset \RR^{n+1}.
\]
By the chain rule, for any $1 \leq j \leq n$ and $\zz \in \RR^n$,
\[
 \frac{\partial \tilde{\gg}}{\partial x_j}(\bphi(\zz),\zz) = \frac{\partial 
\gg}{\partial x_{\sigma(j)}}(\sigma^{-1}\cdot \zz)
 \et 
 \frac{\partial \tilde{\bphi}}{\partial x_j}(\bphi(\zz),\zz) = \frac{\partial 
\bphi}{\partial x_{\sigma(j)}}(\sigma^{-1}\cdot \zz).
\]
Hence, for $\jac(\ff,i)$ the Jacobian matrix of $\ff$ with respect to 
$(x_{i+1},\dotsc,x_n)$, and $\tilde{\yy} = (\bphi(\yy),\sigma \cdot \yy)$,
\[
\jac_{\tilde{\yy}}(\hh,d-1) = 
\left[
\begin{array}{cc}
 \jac_{\tilde{\yy}}(\tilde{\gg},d-1)\\
 \jac_{\tilde{\yy}}(\tilde{\bphi},d-1)
\end{array}
\right]
=
\left[
\begin{array}{c}
\frac{\partial \gg}{\partial x_{\sigma(i)}}(\yy)\\
\frac{\partial \bphi}{\partial x_{\sigma(i)}}(\yy) 
\end{array}
\right]_{d \leq j \leq n},
\]
which is invertible by assumption on $\sigma$. 

Therefore, applying the semi-algebraic implicit function theorem \cite[Th 
3.30]{BPR2006} to $\hh$, there is an open Euclidean neighborhoods
$N'\subset \RR^{d}$ of $(\bphi(\yy),\yy')$ where 
$\yy'=(\yy_{\sigma(\ell)}, \, 1\leq \ell \leq d-1) $, an open Euclidean 
neighborhood $N'' \subset \RR^{n-d+1}$ of $\yy'' = (\yy_{\sigma(\ell)},\,d\leq 
\ell \leq n)$ and a map $\ff = (f_1, \dotsc, f_{n-d+1}) \in 
\Spaz^\infty(N',N'')$ (since $\bphi$ and the $g_i$'s are polynomials) such that:
\[
     \forall \, \zz = (\zz', \zz'') \in N' \times N'', \: \Big[ \hh(\zz)=0 
\Longleftrightarrow \zz'' = \ff(\zz') \Big] 
\]
Then, let $N = (N'\times N'') \cap (\RR\times\sigma\cdot \Opaz_{\yy}) \subset 
\RR^{n+1}$, the previous assertion becomes:
\begin{equation}\label{eqn:paramalpha}
    \big\{(\bphi(\zz),\sigma\cdot\zz) \mid \zz \in Z\big\} \cap N =
    \big\{ \left(\zz', \ff(\zz')\right) \: \mid \: \zz' \in N' \big\}
\end{equation}
Finally, we claim that taking $\balpha\colon \zz \in Z \mapsto
(\bphi(\zz),\sigma\cdot\zz)$ ends the proof. Indeed, by
equation~\eqref{eqn:paramalpha}, assertion $a)$ immediately holds since $N'$ and
$N$ are Euclidean open neighborhood of $\pi_d(\balpha(\yy))$ and $\balpha(\yy)$
respectively. Further, one checks that $\balpha$ is a Zariski isomorphism, of
inverse $\sigma^{-1}$ after projecting on the last $n$ coordinates, which proves
$b)$. Finally, one sees that \rev{$\pi_1 \circ \balpha = \bphi$} so that $c)$
holds as well.
\end{proof}
\begin{remark}
  The previous lemma shows in particular that $Z\cap\RR^n-K(\bphi,Z)$ is a Nash
  manifold (see \cite[Section 3.4]{BPR2006}) of dimension $d$, {\it i.e.}\  
locally
  $\Spaz^\infty$-diffeomorphic to $\RR^d$.
\end{remark}

\begin{lemma}\label{lem:voisinageSAC}
Let $\yy$ be in $Z \cap \RR^n - K(\bphi, Z)$ and $u = \bphi(\yy)$. 
Then there exists an open Euclidean neighborhood $N(\yy)$ of $\yy$ such that 
the following holds:
\begin{enumerate}[label=\alph*)]
 \item $N(\yy)$ is \SAC;
 \item $(Z \cap N(\yy))\rinf{u}$ is non-empty and \SAC;
 \item $(Z \cap N(\yy))\req{u}$ is contained in
 \rev{$\bar{\left(Z \cap N(\yy)\right)\rinf{u}}$}.
\end{enumerate}
\end{lemma}
This result is illustrated by Figure~\ref{fig:voisinageSAC}.

\begin{figure}[h]\centering
\begin{minipage}[c]{0.5\linewidth}
 \includegraphics[width=\linewidth]{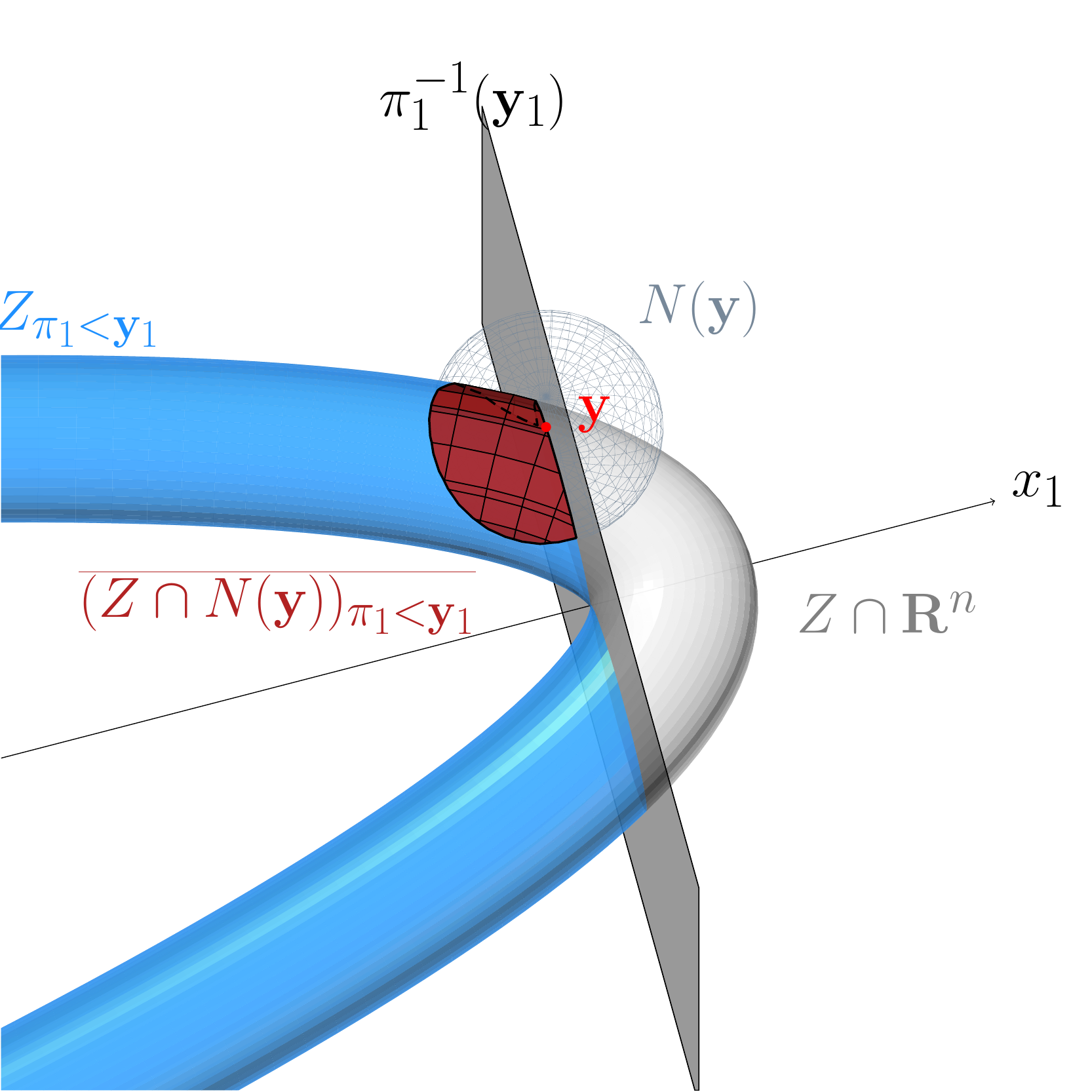}
\end{minipage}\hfill
\begin{minipage}[c]{0.5\linewidth}
 \includegraphics[width=\linewidth]{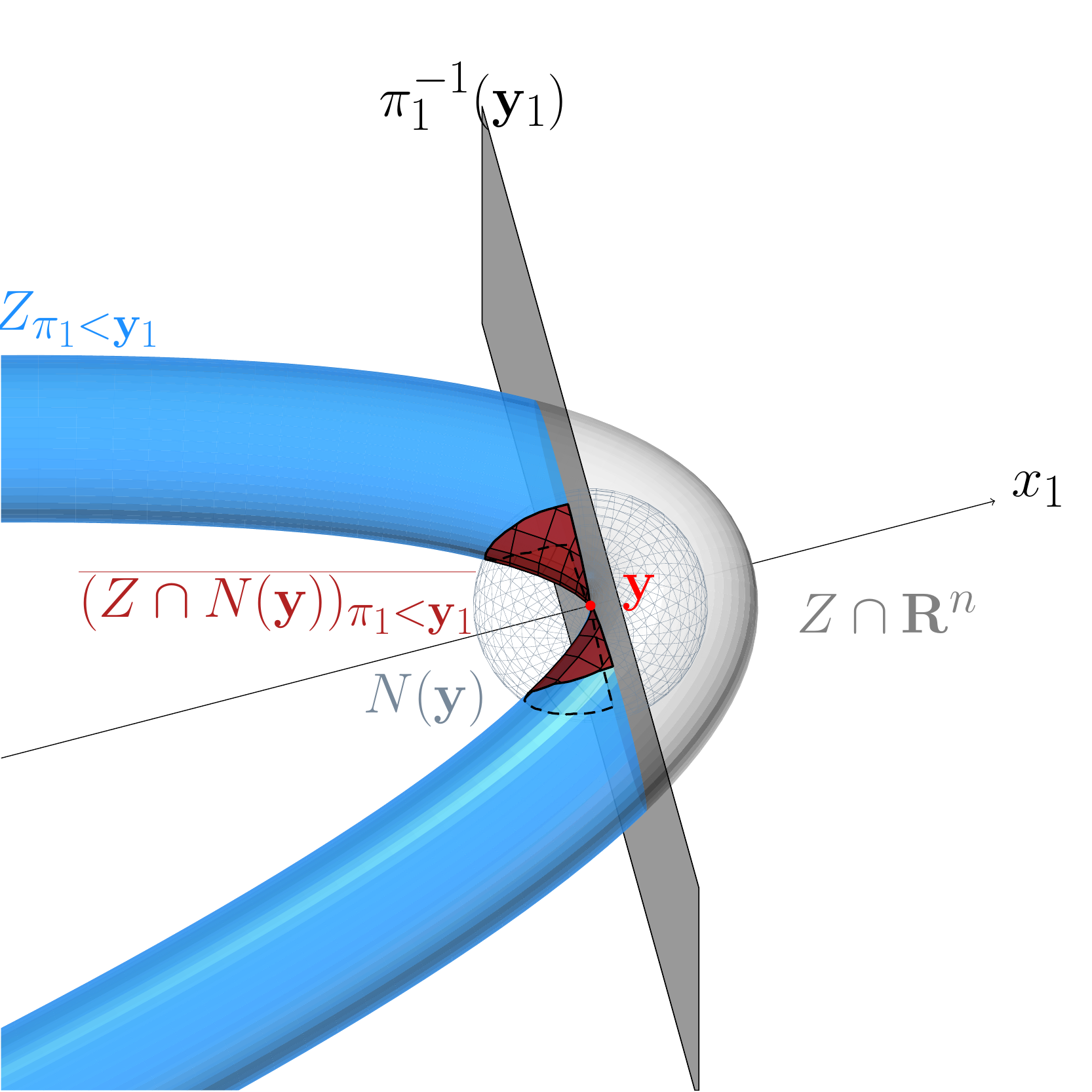}
\end{minipage}
\caption{Illustration of Lemma~\ref{lem:voisinageSAC} where $\bphi = 
\pi_1$, $u=\yy_1$ and $Z$ is isomorphic to $\V(x_1^2+x_2^2-1)\times 
\rev{\V(x_1+x_2^2)}$. 
On the left, $\yy$ is not critical and one sees that it satisfies all the 
statements. 
On the right $\yy$ is critical, and $(Z\cap N(\yy))\rinf[\pi_1]{\yy_1}$ is 
disconnected.
Note that in both cases, $\yy_1$ is a critical \emph{value}.}
\label{fig:voisinageSAC}
\end{figure}

\begin{proof}
Let $\balpha, N',N$ and $\ff$ be obtained by applying Lemma 
\ref{lem:implicitparam}. 
Let $\FF\colon \zz' \in N' \mapsto \left(\zz', \ff(\zz')\right) \in N$.
Let $\eps > 0$ be such that
\[
\Bpaz = \mathcal{B}\left(\pi_d(\alpha(\yy)),\eps \right) \subset N' \subset 
\RR^d 
\]
where $\mathcal{B}\left(\pi_d(\alpha(\yy)),\eps \right)$ is the open ball of 
$\RR^d$ with radius $\eps$ and center $\pi_d(\alpha(\yy))$.
We claim that taking $N(\yy) = \balpha^{-1}(\FF(\Bpaz))$ is enough to prove 
the result.

First, $\FF(\Bpaz)$ is open, semi-algebraic and \SAC, since $\FF$ is an open 
continuous map on $\Bpaz$.
Then, by assumptions on $\balpha$, together with Lemma~\ref{lem:isoconnected}, 
$\balpha ^{-1}(\FF(\Bpaz))$ is a \SAC open neighborhood of $\yy$.
Hence $N(\yy)$ satisfies statement $a)$.

Besides, remark that $\FF(\Bpaz) \subset \balpha(Z)$,
so that
\[
( \balpha(Z) \cap \FF(\Bpaz))\rinf[\pi_1]{u} = 
\FF(\Bpaz)\rinf[\pi_1]{u} = \FF(\Bpaz\rinf[\pi_1]{u})
\]
as $\pi_1(\FF(\zz')) = \pi_1(\zz')$ for $\zz' \in N'$. 
Since $\pi_1(\alpha(\yy)) = \bphi(\yy) = u$, the semi-algebraic set 
$\Bpaz\rinf[\pi_1]{u}$ is non-empty and \SAC (since $\Bpaz$ is convex), and  
so is its image through $\FF$ by \cite[Section 3.2]{BPR2006}.
But remark that for all $X\subset \RR$,
\begin{equation}\label{eqn:equivalenceiso}
 (Z \cap N(\yy))\rinto{X} 
 = \balpha^{-1}\left(( \balpha(Z) \cap \FF(\Bpaz))\rinto[\pi_1]{X}\right) = 
\alpha^{-1}\circ \FF(\Bpaz\rinto[\pi_1]{X}),
\end{equation}
since $\bphi \circ \balpha^{-1}=\pi_1$.
Therefore, by Lemma~\ref{lem:isoconnected}, $(Z \cap N(\yy))\rinf{u}$ 
is non-empty and \SAC, as claimed in statement $b)$.

To prove assertion $c)$, remark that $\Bpaz\req[\pi_1]{u}$ is contained in 
$\bar{\Bpaz\rinf[\pi_1]{u}}$, so that $\alpha^{-1} \circ 
\FF(\Bpaz\req[\pi_1]{u})$ is contained in $\alpha^{-1} \circ 
\FF(\bar{\Bpaz\rinf[\pi_1]{u}})$. 
Since $\FF$ and $\alpha^{-1}$ are continuous,
\[
\alpha^{-1} \circ \FF\left(\bar{\Bpaz\rinf[\pi_1]{u}}\right) 
\subset \bar{\alpha^{-1} \circ \FF\left(\Bpaz\rinf[\pi_1]{u}\right)}.
\]
Finally, by~\eqref{eqn:equivalenceiso}, we get
\[
(Z \cap N(\yy))\req{u} \subset
\bar{(Z \cap N(\yy)\rinf{u}}.
\]
\end{proof}
\begin{figure}[h]\centering
\begin{minipage}[c]{0.5\linewidth}
 \includegraphics[width=\linewidth]{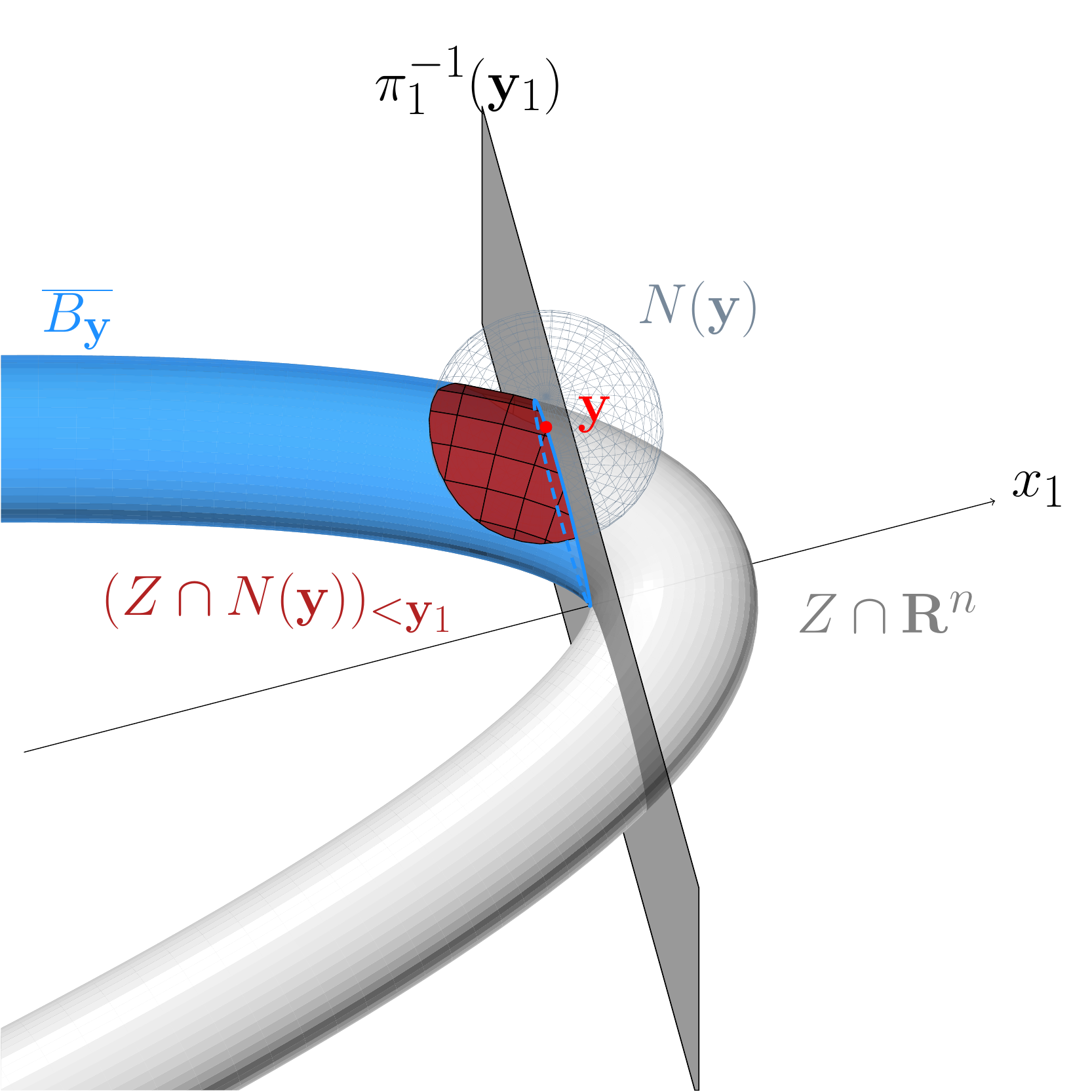}
\end{minipage}\hfill
\begin{minipage}[c]{0.5\linewidth}
 \includegraphics[width=\linewidth]{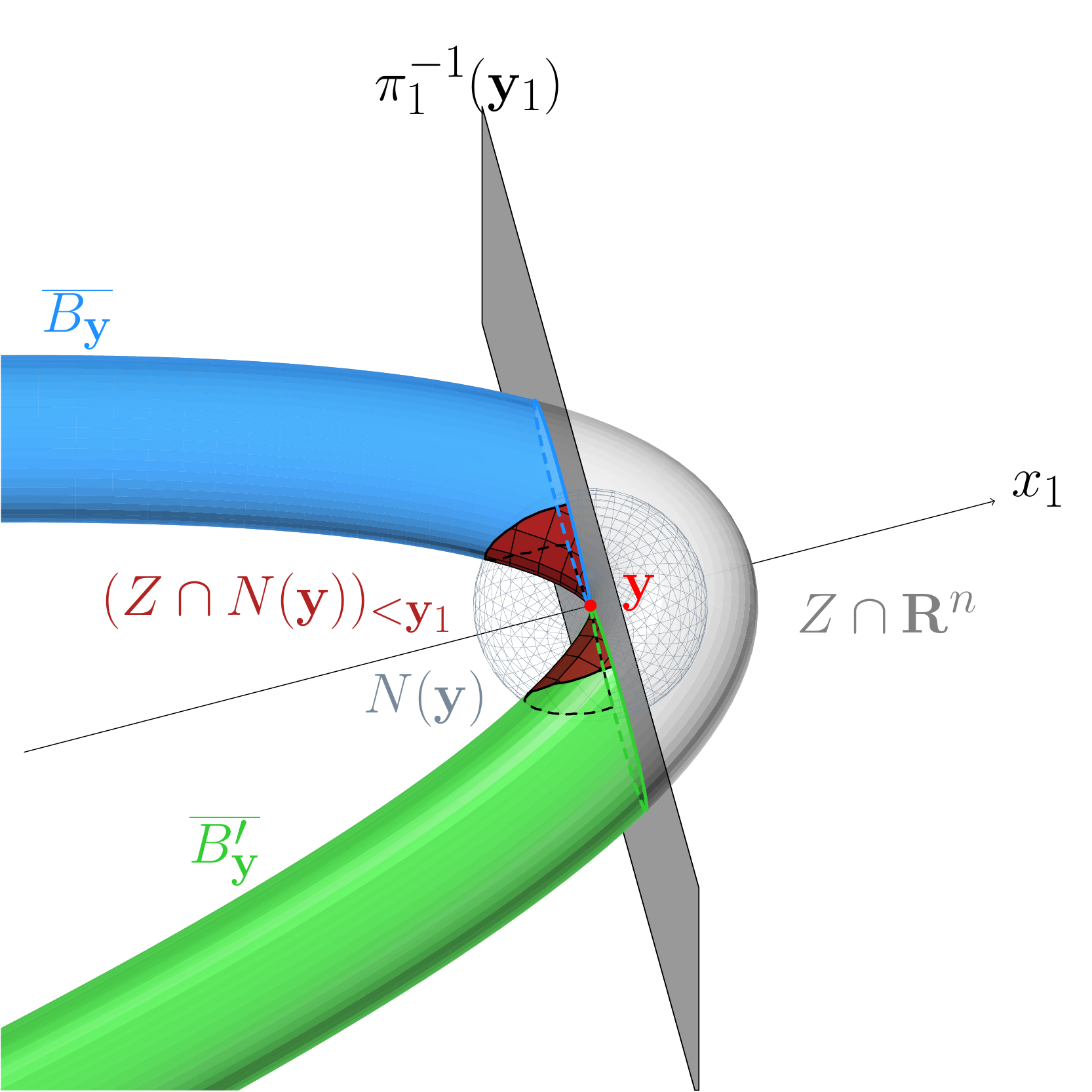}
\end{minipage}
\caption{Illustration of Lemma~\ref{lem:adhcomponent} where $\bphi = 
\pi_1$, $u=\yy_1$ and $Z$ is isomorphic to $\V(x_1^2+x_2^2-1)\times 
\rev{\V(x_1+x_2^2)}$. 
On the left $\yy$ is not critical and one sees that $\yy \in \bar{B_\yy}$ and 
$(Z \cap N(\yy))\rinf[\pi_1]{\yy_1} \subset B_\yy$.
However on the right, $\yy$ is critical, and one observes that $\yy$ belongs to 
both $\bar{B_\yy}$ and $\bar{B_\yy'}$, and, in addition, that $(Z \cap 
N(\yy))\rinf[\pi_1]{\yy_1}$ is not contained in any of these components.
Note that in both cases, $\yy_1$ is a critical \emph{value}.}
\label{fig:adhcomponent}
\end{figure}
\begin{lemma}\label{lem:adhcomponent} 
Let $\yy$ be in $Z \cap \RR^n - K(\bphi,Z)$, let $u=\bphi(\yy)$ and let 
$N(\yy)$ as in Lemma~\ref{lem:voisinageSAC}.
Then, there exists a \emph{unique} \SACC $B_\yy$ of $Z\rinf{u}$ such that $\yy 
\in \bar{B_\yy}$. 
Moreover,
\[
 (Z \cap N(\yy))\rinf{u} \subset B_\yy.
\]
\end{lemma}
This lemma is illustrated in Figure~\ref{fig:adhcomponent}.
\begin{proof}
By the second item of Lemma~\ref{lem:voisinageSAC}, $(Z \cap N(\yy))\rinf{u}$ 
is non-empty and \SAC. 
Thus, it is contained in a \SACC $B_\yy$ of $Z\rinf{u}$.
Since the \SACCs of $Z\rinf{u}$ are pairwise disjoint, $B_\yy$ is well defined 
and unique.
Moreover by Lemma~\ref{lem:voisinageSAC},
\[
\yy \in \bar{(Z \cap N(\yy))\rinf{u}} \subset \bar{B_\yy}.
\]
Finally, suppose that there exists another connected component $B'$ of 
$Z\rinf{u}$ such that $\yy \in \overline{B'}$. 
Then $\yy$ belongs to the closure of $B'$, so that $N(\yy) \cap B' \neq 
\emptyset$, since $N(\yy)$ is a neighborhood of $\yy$. 
Thus $B' \cap B_\yy$ is not empty, and since they are both {\SACC}s of the same 
set, $B' = B_\yy$.
\end{proof}

Let us see a geometric consequence of this result.
The following lemma shows that if $u$ is the least element of $\RR$ such that 
the hypersurface $\bphi^{-1}(\{u\})$ intersects a \SACC $C$ of $Z \cap \RR^n$, 
then this intersection consists entirely of singular points of $\bphi$ on $Z$.
It is illustrated by Figure~\ref{fig:SACCvide}.
\begin{lemma}\label{lem:SACCvide}
 Let $\yy \in Z\cap\RR^n$ with $u=\bphi(\yy)$ and let $C$ be the \SACC of 
$Z\rinfeq{u}$ containing $\yy$.
 If $C\rinf{u} = \emptyset$ then $C = C\req{u} \subset K(\bphi,Z)$.
 In particular, $\yy \in K(\bphi,Z)$.
\end{lemma}

\begin{figure}[h]\centering
\begin{minipage}[c]{0.5\linewidth}
 \includegraphics[width=\linewidth]{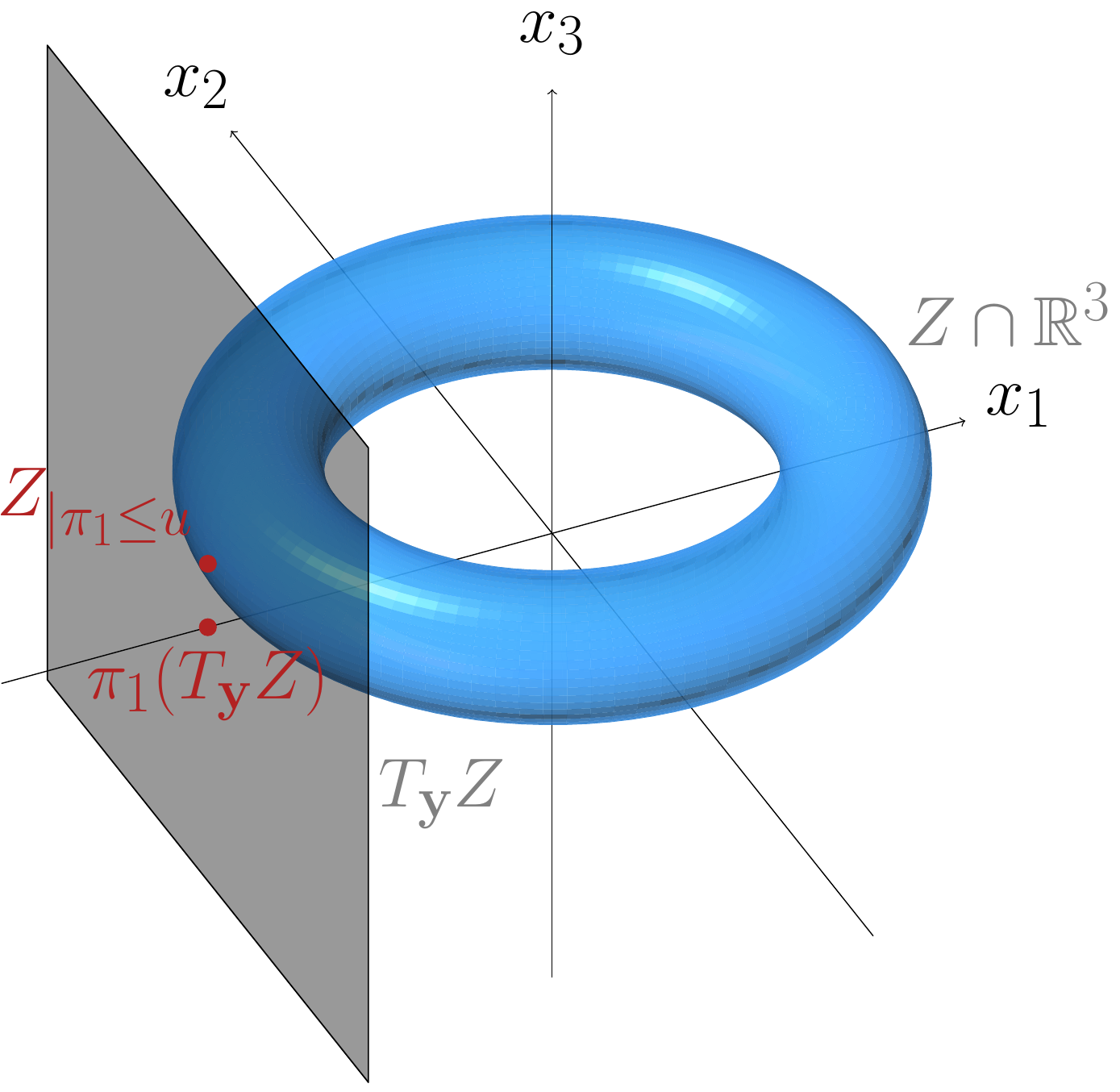}
\end{minipage}\hfill
\begin{minipage}[c]{0.5\linewidth}
 \includegraphics[width=\linewidth]{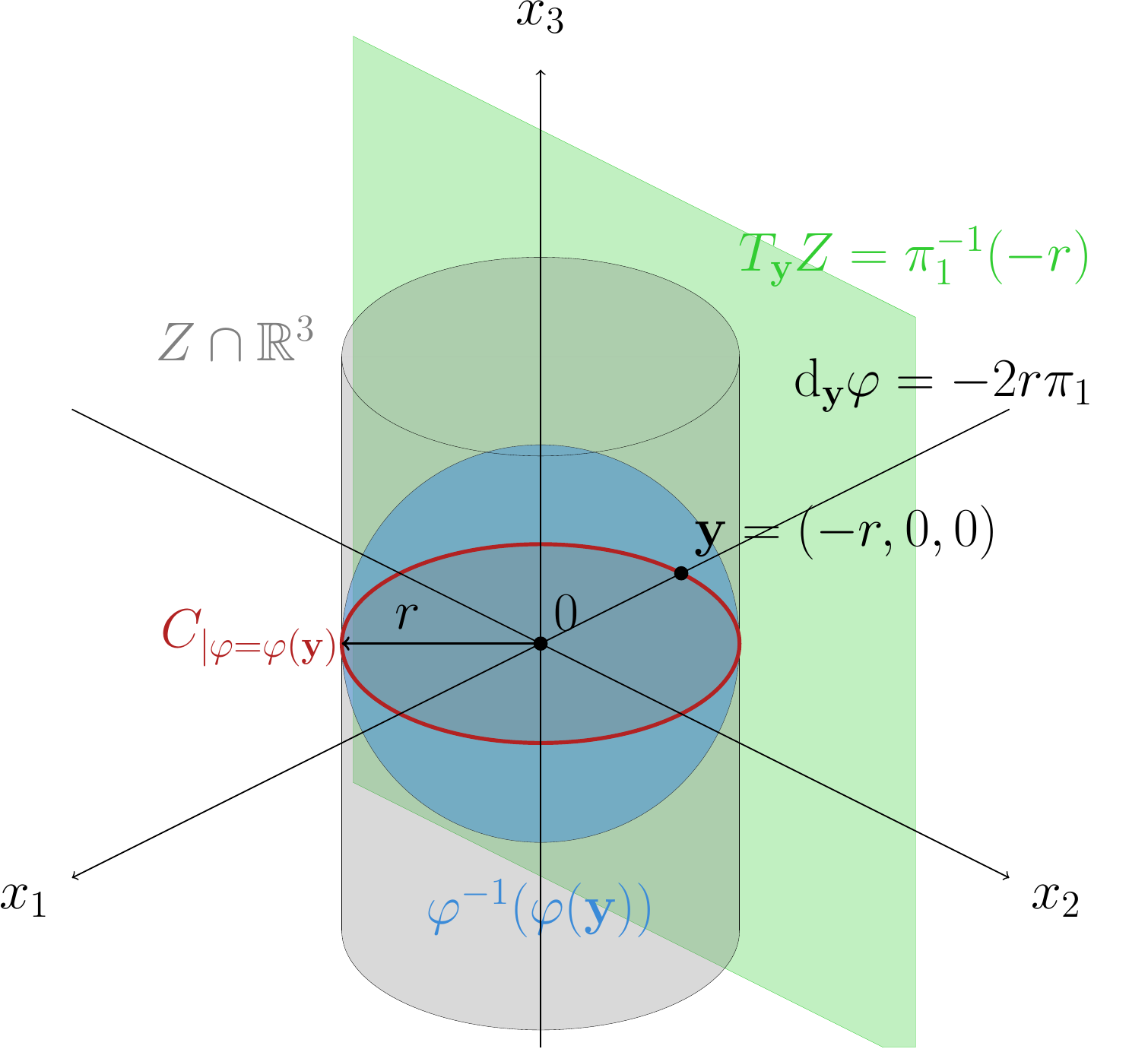}
\end{minipage}
\caption{Illustration of Lemma~\ref{lem:SACCvide} in two cases. 
\revbis{On the left, $\bphi = \pi_1$ and $Z\cap\RRo^3$ is a torus. 
The plane $\{x_1 = u\}$ indicated satisfies $C\rinf{u} = \emptyset$.
One sees that $C\req{u} \subset K(\bphi,Z)$, and indeed $C\req{u} = \{\yy\}$.
On the \revbis{right}, $\bphi$ is the square of the Euclidean norm, and $Z$ is 
a cylinder of radius $r$.
Remark first that $C\rinf{r} = \emptyset$.
Moreover, for $\xx=(\xx_1,\xx_2,0) \in Z$, the differential at $\xx$ of 
restriction of $\phi$ to $Z$ is the restriction of the projection on the 
$(x_1,x_2)$-plane to $T_{\xx}Z$. Since these two latter planes are orthogonal, 
$\xx$ is indeed a critical point.}}
\label{fig:SACCvide}
\end{figure}

\begin{proof}
If $C\rinf{u}=\emptyset$, since $C \subset Z\rinfeq{u}$ then $C = C\req{u}$ 
holds. 
Let us prove the contrapositive of the rest of the lemma.
Suppose that $C\req{u} \not\subset K(\bphi,Z)$, and let
\[
\zz \in C\req{u}-K(\bphi,Z).
\]
Let $B_{\zz}$ be the \SACC of $Z\rinf{u}$ obtained by applying Lemma 
\ref{lem:adhcomponent}. 
Since $\bar{B_{\zz}}$ contains $\zz$ and is a \SAC set of $Z\rinfeq{u}$, 
$\bar{B_{\zz}} \subset C$. 
Hence $C\rinf{u}$ contains $(\bar{B_{\zz}})\rinf{u} = B_{\zz}$,  which is then 
not empty.
\end{proof}

We prove now an important consequence of the previous lemma.
It is a fundamental property of generalized polar varieties and motivates 
their introduction among the ingredients of a roadmap.
\begin{proposition}\label{prop:SACCcontientptcrit}
 Let $u \in \RR$ and let $B$ be a \emph{bounded} \SACC of $Z\rinf{u}$.
 Then $B \cap K(\bphi,Z) \neq \emptyset$.
\end{proposition}
\begin{proof}
Since $\bphi$ is a semi-algebraic continuous map and $B$ is semi-algebraic, then
$\bphi(\bar{B})$ is a closed and bounded semi-algebraic set by \cite[Theorem 
3.23]{BPR2006}.
In particular, $\bphi$ reaches its minimum $\bphi(\zz)$ on $\bar{B}$ and since 
$\emptyset \neq B\subset Z\rinf{u}$,then $\bphi(\zz) <u$.
But $B$ is a \SACC of $Z\rinf{u}$, so in particular it is closed in 
$Z\rinf{u}$, 
so that 
\[
 \bar{B} - B\subset Z\req{u}.
\]
Therefore $\zz \in B$ and as $B\rinf{\bphi(\zz)}$ is empty ($\zz$ is a 
minimizer), $B\req{\bphi(\zz)}$ and $\zz$ is in $K(\bphi, Z)$ by 
Lemma~\ref{lem:SACCvide}.
Finally $\zz \in B \cap K(\bphi,Z)$, and the latter is non-empty.
\end{proof}

We are now able to prove a weaker version of 
Proposition~\ref{prop:firstresult}, which is illustrated in 
Figure~\ref{fig:firstresultweak}.
It deals with the particular case when the map has values in some fiber 
$Z\req{u}$, where $u \in \RR$.
\begin{lemma}\label{lem:firstresultweak}
Let $u\in\RR$ and $A \subset \RR^k$ be a \SAC set. Let
\[
 \gamma\colon A \longrightarrow Z\req{u} -K(\bphi,Z)
\]
be a continuous semi-algebraic map. Then there exists a \emph{unique} \SACC $B$ 
of $Z\rinf{u}$ such that $\gamma(A) \subset \bar{B}$.
\end{lemma}

\begin{figure}[h]\centering
\begin{minipage}[c]{0.49\linewidth}
 \includegraphics[width=\linewidth]{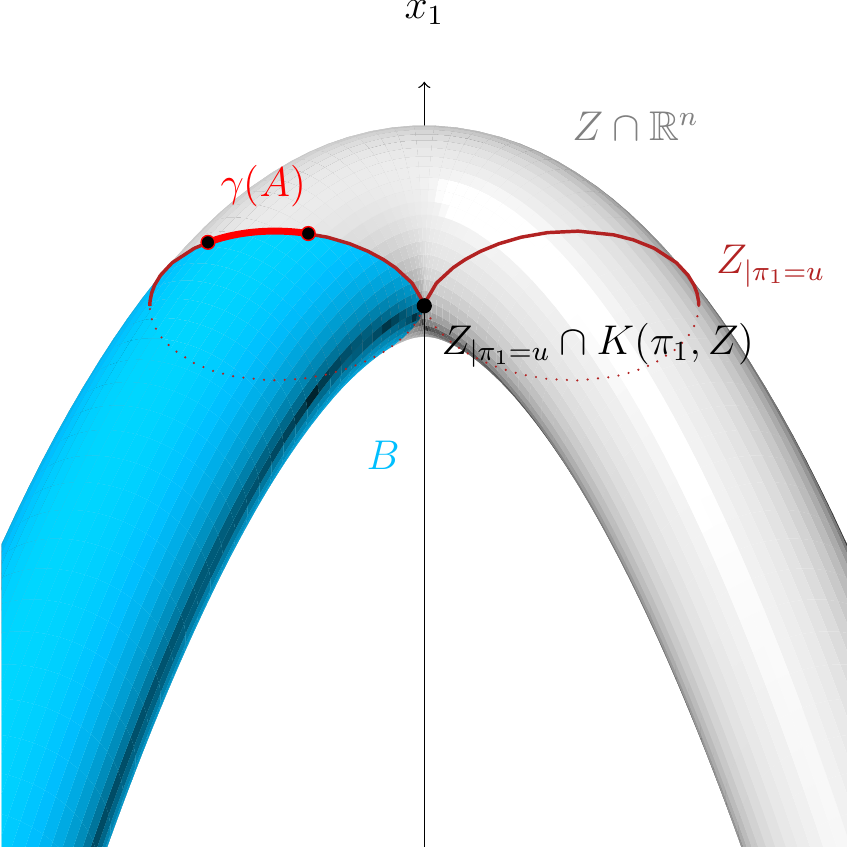}
\end{minipage}\hfill
\begin{minipage}[c]{0.49\linewidth}
\includegraphics[width=\linewidth]{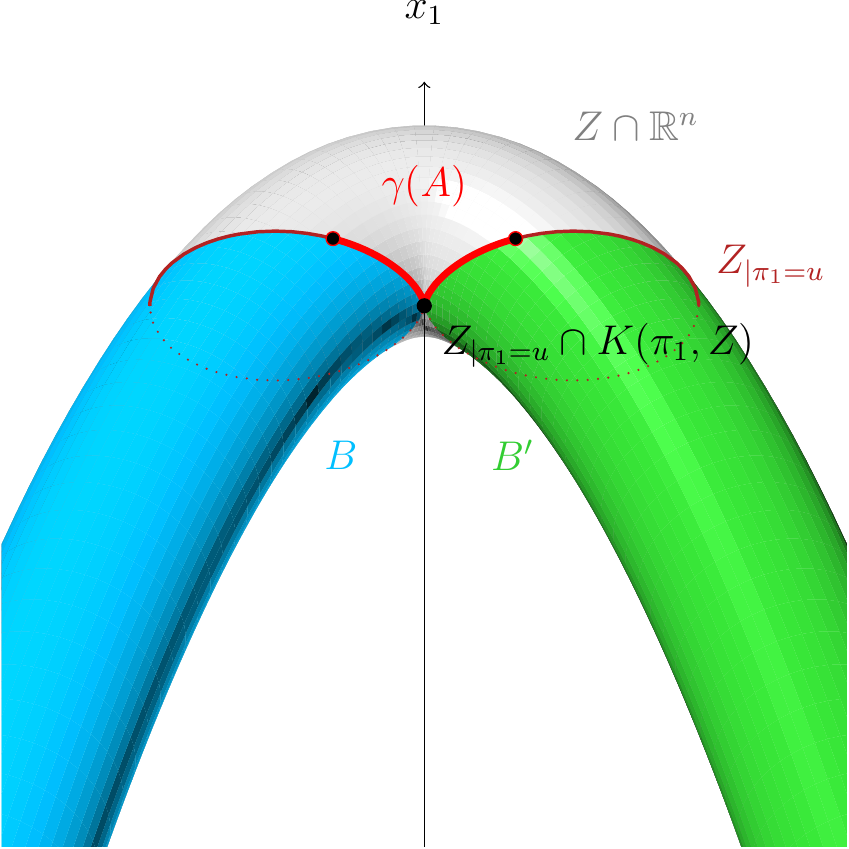}
\end{minipage}
\caption{Illustration of the proof of Proposition~\ref{prop:firstresult} where 
$\bphi = \pi_1$ and $Z$ is isomorphic to $\V(x_1^2+x_2^2-1)\times 
\rev{\V(x_1+x_2^2)}$ in two cases. 
On the left the $\gamma(A) \cap (Z\req[\pi_1]{u} \cap K(\pi_1,Z)) = 
\emptyset$ and on the right, this intersection is non-empty.}
\label{fig:firstresultweak}
\end{figure}

\begin{proof}
Let $\aa_0 \in A$ and $\yy = \gamma(\aa_0)$, by assumption, $\yy \in 
Z\req{u}-K(\bphi,Z)$.
Then by Lemmas~\ref{lem:voisinageSAC} and \ref{lem:adhcomponent},
there exist an open neighborhood $N(\yy)$ of $\yy$ and a \SACC $B_\yy$ of 
$Z\rinf{u}$ such that
\[
 (Z \cap N(\yy))\req{u} \subset
 \bar{(Z \cap N(\yy))\rinf{u}} \subset \bar{B_{\yy}}. 
\]
Hence for every $\zz \in (Z \cap N(\yy))\req{u}  - K(\bphi,Z)$, $\zz \in 
\bar{B_{\yy}}$ so that $B_{\zz} = B_{\yy}$ by application of 
Lemma~\ref{lem:adhcomponent}.
Since $\gamma$ is a continuous semi-algebraic map, there exists an open 
semi-algebraic neighborhood $N'(\aa_0)$ of $\aa_0$ such that 
\[
\gamma(N'(\aa_0)) \subset (Z \cap N(\yy))\req{u}  - K(\bphi,Z).
\]
\rev{
Hence the map $\aa \mapsto B_{\gamma(\aa)}$ is constant on 
$N(\aa_0)$.
Let 
\[
 \Bfrak: \aa\in A \mapsto B_{\gamma(\aa)} \in \Ppaz(Z\rinf{u})
\]
be the map given by Lemma~\ref{lem:adhcomponent}, where $\Ppaz(Z\rinf{u})$ 
denote the power set of $Z\rinf{u}$.
We proved that $\Bfrak$ is locally constant on $A$ and then, equivalently, 
continuous for the discrete topology on $\Ppaz(Z\rinf{u})$.
But since $A$ is \SAC, $\Bfrak(A)$ is connected for the discrete topology, that 
is $\Bfrak$ is constant $A$.}

\rev{Let then $B$ be the constant value that $\Bfrak$ takes on $A$.}
By Lemma~\ref{lem:adhcomponent}, for all $\aa \in A$, $\gamma(\aa) \in 
\overline{B_{\gamma(\aa)}} = \overline{B}$, that is $\gamma(A)\subset\bar{B}$.
Besides, if $B'$ is another \SACC of $Z\rinf{u}$ such that $\gamma(A) \subset 
\bar{B'}$, then for all $\aa \in A$,
\[
 \gamma(\aa) \in \bar{B} \cap \bar{B'} \cap Z\req{u} - K(\bphi,Z),
\]
so that $B = B'$ by uniqueness in Lemma~\ref{lem:adhcomponent}.
\end{proof}

We can now prove the main proposition by sticking together all the pieces. 
The points of the map that belong to the fiber $Z\req{u}$ are managed by 
Lemma~\ref{lem:firstresultweak}, while the remaining ones, in $Z\rinf{u}$, are 
more convenient to deal with. This proof is illustrated by 
Figure~\ref{fig:firstresult}.

\begin{proof}[Proof of Proposition~\ref{prop:firstresult}]
Since $\gamma$ is semi-algebraic and continuous, $\gamma(A)$ is \SAC.
Hence, if $\gamma(A) \subset Z\rinf{u}$, it is contained in a unique \SACC $B$ 
of $Z\rinf{u}$ and we are done.

We assume now that $\gamma(A) \not\subset Z\rinf{u}$. 
Let $G = \gamma^{-1}(Z\req{u})$.
It is a closed subset of $A$ since $Z\req{u}$ is closed in $Z\rinfeq{u}$ and 
$\gamma$ is continuous. 
Then, let $G_1,\dotsc,G_N$ be the \SACCs of $G$; they are closed in $A$ 
since they are closed in $G$, which is closed in $A$.
Besides, let $H_1,\dotsc, H_M$ be the \SACCs of $A-G$. They are open in $A$ 
since they are open in $A-G$ , which is open in $A$.

We define a map $\Bfrak\colon A \to \Ppaz(Z\rinf{u})$, where $\Ppaz(Z\rinf{u})$ 
is
the power set of $Z\rinf{u}$. 
The family formed by both $G_1,\dotsc,G_N$ and $H_1,\dotsc H_M$
is a partition of $A$; hence, we can define $\Bfrak$ by defining it on this 
partition.
\begin{description}
 \item[$\bm{H_i:}$] Since $H_i \subset A-G$, $\gamma(H_i) \subset 
Z\rinf{u}$ and $\gamma(H_i)$ is \SAC as $\gamma$ is continuous. 
Then, there exists a unique \SACC $B_i$ of $Z\rinf{u}$ such that 
$\gamma(H_i) \subset B_i \subset \bar{B_i}$.
   
 \item[$\bm{G_i:}$] Since $G_i$ is \SAC and $\gamma(G_i) \subset Z\req{u} - 
K(\bphi,Z)$, Lemma~\ref{lem:firstresultweak} with $A= G_i$ states that there is 
a unique \SACC $B'_i$ of $Z\rinf{u}$ such that $\gamma(G_i) \subset 
\overline{B'_i}$.
\end{description}
Therefore, for all $\aa \in A$, let $\Bfrak$ such that
\begin{align*}
  \Bfrak(\aa) = \left\{\begin{array}{ll}
            B_i & \text{if $\aa \in H_i$}\\
            B_i' & \text{if $\aa \in G_i$}\\
          \end{array}\right. \text{so that }
  \gamma(\aa) \in \bar{\Bfrak(\aa)}.
\end{align*}
Let us show that $\Bfrak$ is locally constant, that is, for every $\aa \in A$, 
there exists an open Euclidean neighborhood $N(\aa) \subset A$ of $\aa$, such 
that for all $\aa' \in N(\aa)$, $\Bfrak(\aa') = \Bfrak(\aa)$.
Then, we will conclude by connectedness \rev{as above}.
Let $\aa \in A$ and $1\leq i \leq\max(M,N)$.
\begin{itemize}
\item If $\aa \in H_i$, since $H_i$ is open in $A$, there exists an open 
Euclidean neighborhood $N(\aa)$ of $\aa$ contained in $H_i$. 
By construction, for all $\aa' \in N(\aa)$, $\Bfrak(\aa') = \Bfrak(\aa)$. 
Moreover, since $H_i$ is \SAC, this also proves that $\Bfrak$ is actually 
constant on $H_i$, and we let $\Bfrak(H_i)$ be the unique value it assumes on 
$H_i$.
\item Else $\aa \in G_i$, since the $G_j$'s are closed in $A$, then
$\aa$ does not belong to the closure of any other $G_j$, $j \neq i$.
However, the set 
\[
J = \left\{ 1 \leq j \leq M \mid \aa \in \overline{H_j} \right\}
\] 
is not empty. 
By construction, $\gamma(\aa) \in \overline{\Bfrak(\aa)}$ and by definition of 
$J$, for every $j \in J$, $\gamma(\aa) \in \overline{\Bfrak(H_j)}$. 
But, by Lemma~\ref{lem:adhcomponent} applied with $\yy=\gamma(\aa)$,
such a \SACC is unique. 
Hence for all $j \in J$, $\Bfrak(H_j) = \Bfrak(\aa)$.
One can then take $N(\aa) = \mathcal{B}(\aa,r)$ with $r>0$ such that this open 
ball intersects either the $H_j$'s for $j \in J$ or $G_i$, and only them.
\end{itemize}

\rev{Finally, we proved that $\Bfrak$ is locally constant and then, 
equivalently, continuous for the discrete topology on $\Ppaz(Z\rinf{u})$.
Since $A$ is \SAC, $\Bfrak(A)$ is connected for the discrete topology and
$\Bfrak$ is constant on $A$.}
Denoting by $B\subset Z\rinf{u}$ the unique value it assumes, we have 
$\gamma(A) 
\subset \overline{B}$ as claimed.
Besides if $B'$ is another \SACC of $Z\rinf{u}$ such that $\gamma(A) \subset 
\bar{B'}$, then in particular $\bar{B} \cap \bar{B'}$ contains $\gamma(G_1) 
\subset Z\req{u}-K(\bphi,Z)$, so that $B = B'$ by 
Lemma~\ref{lem:firstresultweak}.
\end{proof}

\begin{figure}[h]\centering
\begin{minipage}[c]{0.49\linewidth}
 \includegraphics[width=\linewidth]{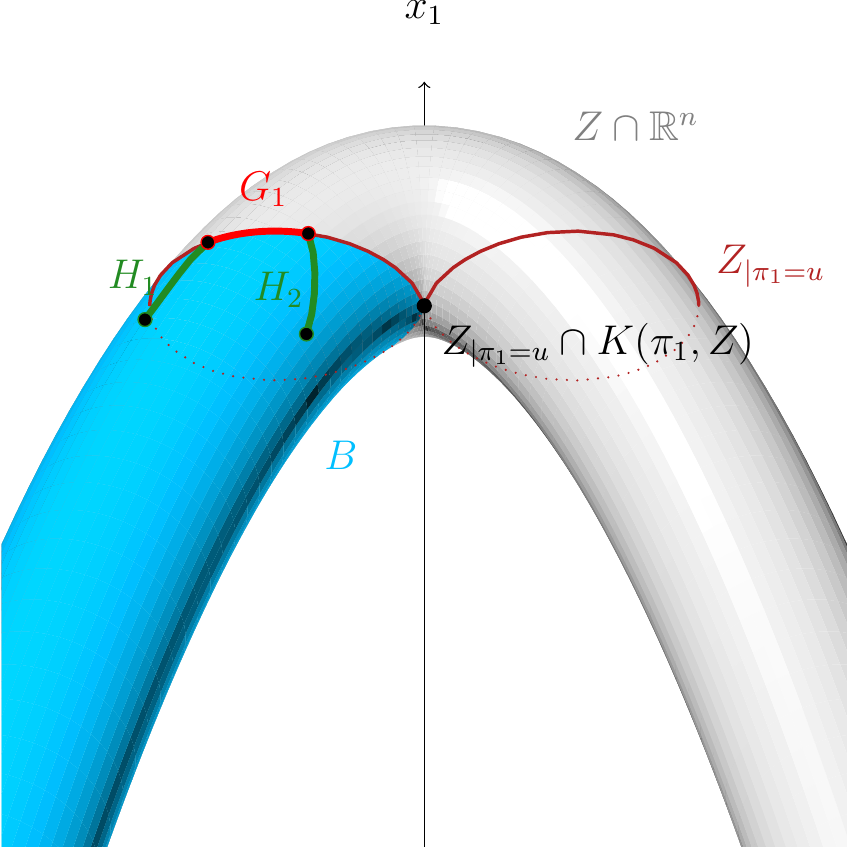}
\end{minipage}\hfill
\begin{minipage}[c]{0.49\linewidth}
 \includegraphics[width=\linewidth]{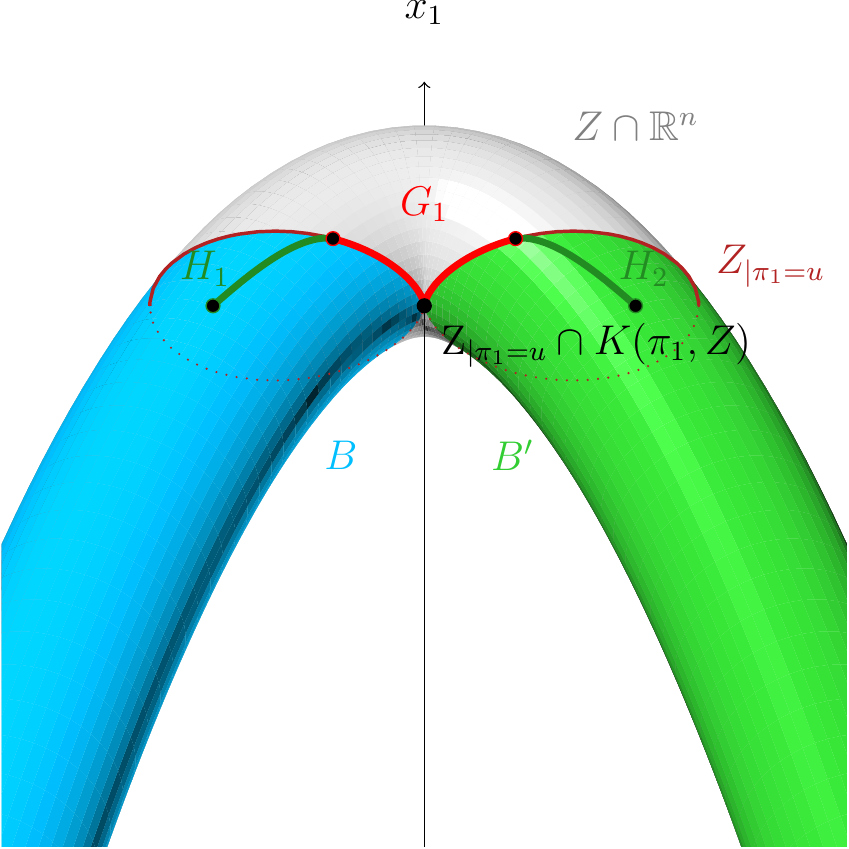}
\end{minipage}
\caption{Illustration of the proof of Proposition~\ref{prop:firstresult} with 
$\bphi = \pi_1$ and $Z$ is isomorphic to $\V(x_1^2+x_2^2-1)\times 
\rev{\V(x_1+x_2^2)}$ in two cases. 
The intersection $\gamma(A) \cap (Z\req[\pi_1]{u} \cap K(\pi_1,Z))$ is 
empty on the left while, on the right, it is not.}
\label{fig:firstresult}
\end{figure}

We then deduce the following consequence on the \SACCs of $Z$ with respect to 
$\bphi$.
This result is illustrated in Figure~\ref{fig:firstresultcor}.
\begin{corollary}\label{cor:firstresult}
Let $\bphi\colon \CC^n \rightarrow \CC$ be a regular map defined over
$\RR$ and $Z \subset \CC^n$ be an equidimensional algebraic set of positive 
dimension.
Let $u\in\RR$ such that $Z\req{u} \cap K(\bphi,Z) = \emptyset$ and let $C$ be a 
\SACC of $Z\rinfeq{u}$. 
Then, $C\rinf{u}$ is a \SACC of $Z\rinf{u}$.
\end{corollary}
\begin{proof}
Let $\gamma$ be the inclusion map $\gamma\colon C \hookrightarrow Z\rinfeq{u}$. 
Since $Z\req{u} \cap K(\bphi,Z) = \emptyset$, $\gamma$ satisfies the 
assumptions of Proposition~\ref{prop:firstresult} with $A = C$. 
Then there exists a unique \SACC $B$ of $Z\rinf{u}$ such that $C \subset 
\overline{B}$, so that $C\rinf{u} \subset \overline{B}\rinf{u}=B$. 

First, since $Z\req{u} \cap K(\bphi,Z) = \emptyset$ by assumption, then in 
particular $C\req{u} \not\subset K(\bphi,Z)$. 
By the contrapositive of Lemma~\ref{lem:SACCvide}, $C\rinf{u}$ is not empty.
Hence, since $B$ is a \SAC set of $Z\rinfeq{u}$, containing $C\rinf{u}$, 
$B$ is contained in the \SACC $C$ of $Z\rinfeq{u}$. 
Finally $B \subset Z\rinf{u}\cap C = C\rinf{u}$ and $C\rinf{u} = B$, which is a 
\SACC of $Z\rinf{u}$.
\end{proof}

\begin{figure}[h]\centering
\begin{minipage}[c]{0.5\linewidth}
 \includegraphics[width=\linewidth]{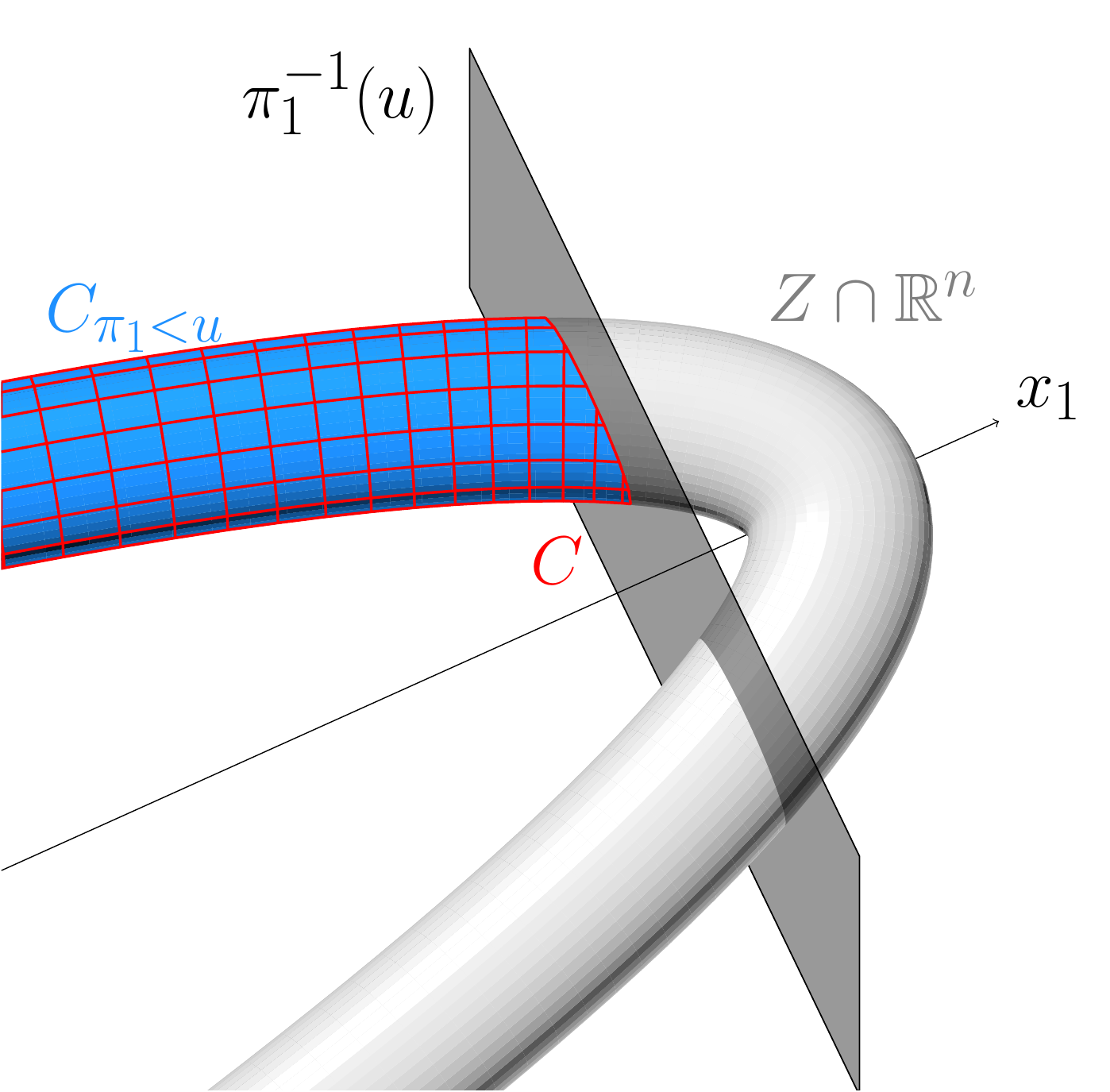}
\end{minipage}\hfill
\begin{minipage}[c]{0.5\linewidth}
 \includegraphics[width=\linewidth]{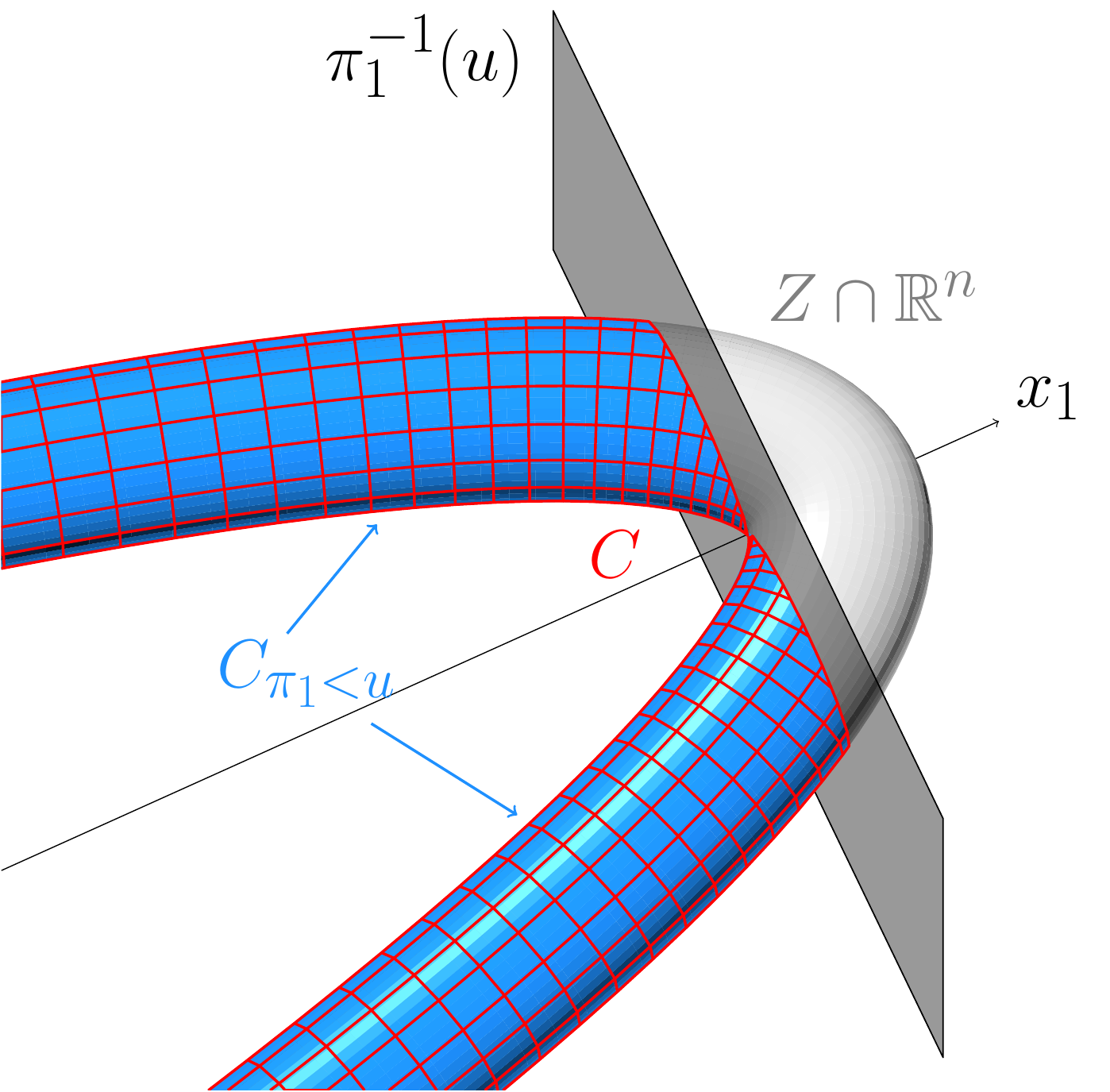}
\end{minipage}
\caption{Illustration of Corollary~\ref{cor:firstresult} where $\bphi = 
\pi_1$ and $Z$ is isomorphic to $\V(x_1^2+x_2^2-1)\times \rev{\V(x_1+x_2^2)}$.
On the left $Z\req[\pi_1]{u} \cap K(\pi_1,Z) = \emptyset$ and one sees that 
$C\rinf[\pi_1]{u}$ is still a \SACC of $Z\rinf[\pi_1]{u}$. 
On the right $Z\req[\pi_1]{u} \cap K(\pi_1,Z) \neq \emptyset$ and one sees that 
$C\rinf[\pi_1]{u}$ is disconnected.}
\label{fig:firstresultcor}
\end{figure}

\subsection{Fibration and critical values}
As in \cite[Section 3.2]{SS2011} we are going to use a Nash version of Thom's
isotopy lemma, stated in \cite{CS1992}\rev{, which, again, is an ingredient of
  Morse theory}. We refer to \cite[Section 3.5]{BPR2006} for the definitions of
Nash diffeomorphisms, manifolds and submersions together with their properties.

\begin{proposition}\label{prop:secondresult}
 Let $\bphi\colon \CC^n \rightarrow \CC$ be a regular map defined over
$\RR$ and $A\subset \bphi^{-1}((-\infty,w)) \cap \RR^n$ be a \SAC 
semi-algebraic set. 
Let $v<w$ such that $A\rinto{(v,w)}$ is a non-empty Nash manifold, bounded, 
closed in $\bphi^{-1}((v, w)) \cap \RR^n$ and such that $\bphi$ is a submersion 
on $A\rinto{(v,w)}$.
Then for all $u \in [v,w)$, $A\rinfeq{u}$ is non-empty and \SAC.
\end{proposition}
\begin{proof}
We first prove that $\bphi\colon A\rinto{(v,w)} \rightarrow (v,w)$ is a proper 
surjective submersion.
Since $A\rinto{(v,w)}$ is bounded and $\bphi$ is semi-algebraic and continuous, 
$\bphi\colon A\rinto{(v,w)} \rightarrow (v,w)$ is a proper map.
Let us prove that $\bphi$ is also surjective on $A\rinto{(v,w)}$ that is 
\[
 \bphi(A\rinto{(v,w)}) = (v,w).
\]

By assumption, $\bphi$ is a submersion from $A\rinto{(v,w)}$ to $(v,w)$.
Then by the semi-algebraic inverse function theorem \cite[Proposition 
3.29]{BPR2006}, $\bphi$ is an open map.
Besides, as $A\rinto{(v,w)}$ is closed and bounded, there exists a closed and 
bounded semi-algebraic set $X \subset \RR^n$ such that $A\rinto{(v,w)} = X \cap 
\bphi^{-1}((v, w)) = X\rinto{(v,w)}$. 
Then 
\[
\bphi(A\rinto{(v,w)}) = \bphi(X\rinto{(v,w)}) = \bphi(X) \cap (v,w).
\]
Since $X$ is bounded and closed, $\bphi(X)$ is closed and bounded by 
\cite[Theorem 3.23]{BPR2006}.
Hence, $\bphi(A\rinto{(v,w)})$ is both open and closed in $(v,w)$. 
Since $(v,w)$ is \SAC, $\bphi(A\rinto{(v,w)}) = (v,w)$.

By the Nash version of Thom's isotopy lemma \cite[Theorem 2.4]{CS1992}, since 
the map $\bphi\colon A\rinto{(v,w)} \rightarrow (v,w)$ is a proper surjective 
submersion, it is a globally trivial fibration.
Hence, for $\zeta \in (v,w)$, there exists a Nash diffeomorphism $\Psi$ of the 
form
\[
\begin{array}{cccc} 
\Psi\colon & A\rinto{(v,w)} &\longrightarrow & (v,w) \times A\req{\zeta} 
\\[0.3em]
& \yy & \longmapsto & (\;\bphi(\yy)\;\:,\;\,\psi(\yy)\;\;)
\end{array}.
\]
We now proceed to prove the main statement of the proposition. 
There are, at first sight, two different situations to consider:
whether $u>v$ or $u=v$ (see Figure~\ref{fig:secondresult}).
Using Puiseux series, we actually prove them simultaneously.

Take $u\in[v,w)$; we prove that $A\rinfeq{u}$ is non-empty and \SAC. To prove 
that $A\req{u}$ is non-empty, we consider $\zz \in A\req{\zeta}$ and the map
\[
    \begin{array}{cccc} 
    \gamma\colon & [0,1) &\rightarrow & A\rinto{(v, w)}\\[0.3em]
    & t & \mapsto & \Psi^{-1}(tu + (1-t)\zeta, \zz).
    \end{array}
\]
This map is well defined and continuous, since $\Psi$ is a Nash diffeomorphism 
from  $A\rinto{(v,w)}$ to $(v,w) \times A\req{\zeta}$, and satisfies 
$\bphi(\gamma(t)) = tu + (1-t)\zeta$ for every $t\in [0,1)$.
Moreover $\gamma$ is a bounded map as $A\rinto{(v,w)}$ is bounded by 
assumption. 
Then, by \cite[Proposition 3.21]{BPR2006}, $\gamma$ can be continuously 
extended to $[0,1]$, with $\bphi(\gamma(t)) = tu + (1-t)\zeta$ continuous on 
$[0,1]$, and 
$\bphi(\gamma(1)) = u$. 
Finally $\gamma(1) \in A\rinfeq{u} $ and $A\rinfeq{u}$ is not empty.

We prove now that $A\rinfeq{u}$ is \SAC. 
Consider two points $\yy$ and $\yy'$ in $A\rinfeq{u}$. 
Since $A$ is \SAC by assumption, there exists a continuous path $\gamma\colon 
[0,1] \rightarrow A$ such that $\gamma(0) = \yy$ and $\gamma(1) = \yy'$. 
Let us construct, from $\gamma$, another path that lies in $A\rinfeq{u}$. 

Let $\eps$ be an infinitesimal, and let $\RR' = \RPuis$ be the field of 
algebraic Puiseux series in $\eps$ (see \cite[Section 2.6]{BPR2006}). 
We denote by $A', (v,w)',\Psi', \psi', \bphi'$ and $\gamma'$ the extensions of 
respectively $A, (v,w), \Psi, \psi, \bphi$ and $\gamma$ to $\RR'$ in the sense 
of \cite[Proposition 2.108]{BPR2006}.
According to \cite[Exercise 2.110]{BPR2006}, $\Psi'\colon 
A\rinto{(v,w)'}'\rightarrow(v,w)'\times A\req{\zeta}'$ is a bijective map. 
Then let $g'\colon[0,1]'\subset\RR' \rightarrow A'$ be such that
\begin{align*}
    \hspace*{1.5cm}&g'(t) =  \gamma'(t)  
    &\hspace*{-1.5cm}\text{if}&\quad \bphi'(\gamma'(t))\leq u+\eps,\\
    \hspace*{1.5cm}&g'(t) = \Psi'^{-1}(u +\eps, \psi'(\gamma'(t))) 
    &\hspace*{-1.5cm} \text{if}&\quad u + \eps 
    \leq \bphi'(\gamma'(t)) < w.
\end{align*}\\[-1.35em]
This map is well defined since $u+\eps \in (v,w)$ and if $\bphi'(\gamma'(t)) = 
u+\eps$, then $\Psi'^{-1}(u +\eps, \psi'(\gamma'(t))) = \gamma'(t)$. 
Moreover $g'$ is a continuous semi-algebraic map since by \cite[Exercise 
3.4]{BPR2006}, $\Psi'^{-1}$, $\psi'$ and $\gamma'$ are continuous 
semi-algebraic maps.

Finally one observes that $g'$ is bounded over $\RR$.
Indeed if $\bphi'(\gamma'(t))\leq u+\eps$, then $g'(t)=\gamma(t)$, which is 
continuous on $[0,1]'$ and then bounded over $\RR$. 
Else $\bphi'(\gamma'(t)) \in (v,w)$ and $g'(t) \in A'\rinto{(v,w)'}$, which is 
bounded over $\RR$ by \cite[Proposition 3.19]{BPR2006} since 
$A\rinto{(v,w)}$ is. 
Hence, its image $G' = g'([0,1]')$ is a \SAC semi-algebraic set, bounded over 
$\RR$ and contained in $A'\rinfeq{u+\eps}$. 

Let $G = \limeps{G'}$. By \cite[Proposition 12.49]{BPR2006}, $G$ is a closed 
and bounded semi-algebraic set.
Then, since $\bphi$ is a continuous semi-algebraic map defined over $G$, by 
\cite[Lemma 3.24]{BPR2006} for all $\zz' \in G'$,
\[
\bphi( \limeps \zz') = \limeps \bphi(\zz') \leq \limeps{(u + \eps)} = u
\] 
So that $G$ is contained in $A\rinfeq{u}$.
In addition, since $G'$ is \SAC and bounded over $\RR$, then by 
\cite[Proposition 12.49]{BPR2006}, $G$ is \SAC and contains
$\yy = \limeps{g(0)}$ and $\yy' = \limeps{g(1)}$.
 We deduce that there exists, inside $G$, a semi-algebraic path connecting 
$\yy$ to $\yy'$ in $A\rinfeq{u}$, which ends the proof.
\end{proof}

\begin{figure}[h]\centering
\begin{minipage}[c]{0.49\linewidth}
 \includegraphics[width=\linewidth]{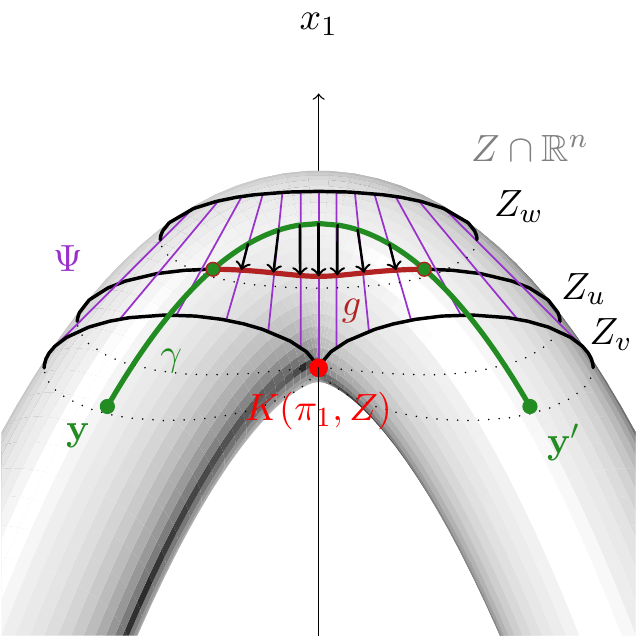}
\end{minipage}\hfill
\begin{minipage}[c]{0.49\linewidth}
 \includegraphics[width=\linewidth]{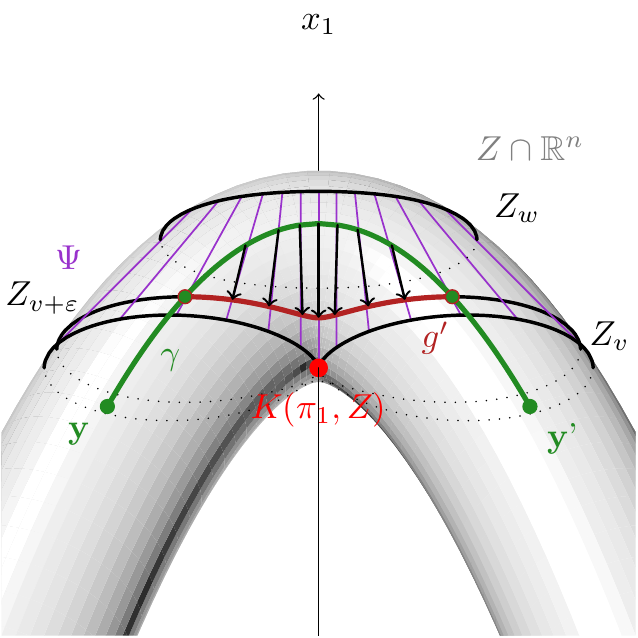}
\end{minipage}
\caption{\rev{Illustration of the two cases covered by the proof of 
Proposition~\ref{prop:secondresult} where $\bphi = \pi_1$ and 
$A=Z\rinf[\pi_1]{w}$, where $Z$ is isomorphic to 
$\V(x_1^2+x_2^2-1)\times \rev{\V(x_1+x_2^2)}$.
The two cases are quite similar; we consider here 
the one where $v$ is a critical value.}
One sees that $\Psi$ connects all the slices $A\req[\pi_1]{u}$ for $u\in 
(v,w)'$. This diffeomorphism allows to transform the problematic parts (not in 
$A\rinfeq[\pi_1]{u}$) of the initial path $\gamma$ (in green), into another 
path $g$ (in red), that lies in $A\req[\pi_1]{u} \subset A\rinfeq[\pi_1]{u}$.}
\label{fig:secondresult}
\end{figure}

The following result is a consequence of Proposition~\ref{prop:secondresult} as 
it deals with a particular case. An illustration of this statement can be found 
in Figure~\ref{fig:secondresultcor}. 
\begin{corollary}\label{cor:secondresult}
Let $Z \subset \CC^n$ be an equidimensional algebraic set of positive dimension 
and let $\bphi\colon \CC^n \rightarrow \CC$ be a regular map defined over $\RR$ 
and proper on $Z\cap \RR^n$.
Let $v<w$ be in $\RR$ such that $Z\rinto{(v,w]} \cap K(\bphi,Z) = \emptyset$, 
and let $C$ be a \SACC of $Z\rinfeq{w}$. 
Then, $C\rinfeq{v}$ is a \SACC of $Z\rinfeq{v}$.
\end{corollary}

\begin{proof}
As $C\rinf{w} = C \cap \bphi^{-1}((-\infty,w)) \cap \RR^n$, we are going to use 
Proposition~\ref{prop:firstresult} with $A = C\rinf{w}$.

First we need to prove that $C\rinf{w}$ is a non-empty \SAC semi-algebraic set. 
Since $Z\req{w} \cap K(\bphi,Z) = \emptyset$, by 
Corollary~\ref{cor:firstresult} $C\rinf{w}$ is a \SACC of $Z\rinf{w}$. 
Hence it is non-empty and \SAC.

Then, we need to prove that $C\rinto{(v,w)}$ is a non-empty Nash manifold, 
bounded and closed in $\bphi^{-1}((v,w)) \cap \RR^n$. 
Suppose first that $C\rinto{(v,w)} = \emptyset$. 
Then 
\[
 C\rinfeq{v} \cup C\req{w} = C \et C\rinfeq{v} \cap C\req{w} = \emptyset.
\]
Since $C$ is \SAC, either $C\rinfeq{v}$ or $C\req{w}$ is empty (as they are 
both closed in $C$). 
In both cases our conclusion follows. 
It remains to tackle the case where $C\rinto{(v,w)}$ is not empty, which we 
assume to hold from now on.

We prove that $C\rinto{(v,w)}$ is bounded. Observe that $C\rinto{(v,w)} 
\subset C\rinto{[v,w]} = C \cap \RR^n \cap \bphi^{-1}([v,w])$. 
Recall that $\bphi$ is proper on $Z \cap \RR^n$ by assumption, and thus on $C 
\cap \RR^n$. Hence, $C\rinto{[v,w]}$ is bounded. 
Besides $C\rinto{(v,w)}$ is closed in $\bphi^{-1}((v,w)) \cap \RR^n$ as
\[
 C\rinto{(v,w)} = C \cap \bphi^{-1}((v,w)) \cap \RR^n,
\] 
and $C$ is closed in $\RR^n$ as it is closed in the closed set $Z\rinfeq{w}$.
Since $C\rinto{(v,w)} \cap K(\bphi,Z) = \emptyset$ then by 
\cite[Proposition 3.3.11]{BCR2013}, $C\rinto{(v,w)}$ is a Nash manifold of 
dimension $\dim(Z)$. 

To apply Proposition~\ref{prop:firstresult}, it remains to prove that $\bphi$ 
is a Nash submersion on $C\rinto{(v,w)}$. 
Let $\yy \in C\rinto{(v,w)}$. Since $\yy \notin \sing(Z)$, then 
$T_\yy C\rinto{(v,w)} = T_\yy Z \cap \RR^n$ according to 
\cite[Proposition 3.3.11]{BCR2013}. 
Since $C\rinto{(v,w)} \cap K(\bphi,Z) = \emptyset$, $d_\yy\bphi$ is onto on 
$T_\yy Z$ and since $\dim{Z}>0$, the image $d_\yy\bphi(T_\yy Z)$  is $\CC$.
Hence 
\[
 d_\yy\bphi(T_{\yy}C\rinto{(v,w)})= \RR.
\]
We just established that all the assumptions of 
Proposition~\ref{prop:secondresult} are satisfied.
One can then apply it to $C\rinf{w}$ and conclude that $C\rinfeq{v}$ is 
non-empty and \SAC.
Finally, since $C$ is a \SACC of $Z\rinfeq{w}$, any \SACC of $Z\rinfeq{v}$ 
contained in $C$ is contained in $C\rinfeq{v}$.
Thus $C\rinfeq{v}$ is a \SACC of $Z\rinfeq{v}$.
\end{proof}

\begin{figure}[h]\centering
\begin{minipage}[c]{0.49\linewidth}
 \includegraphics[width=\linewidth]{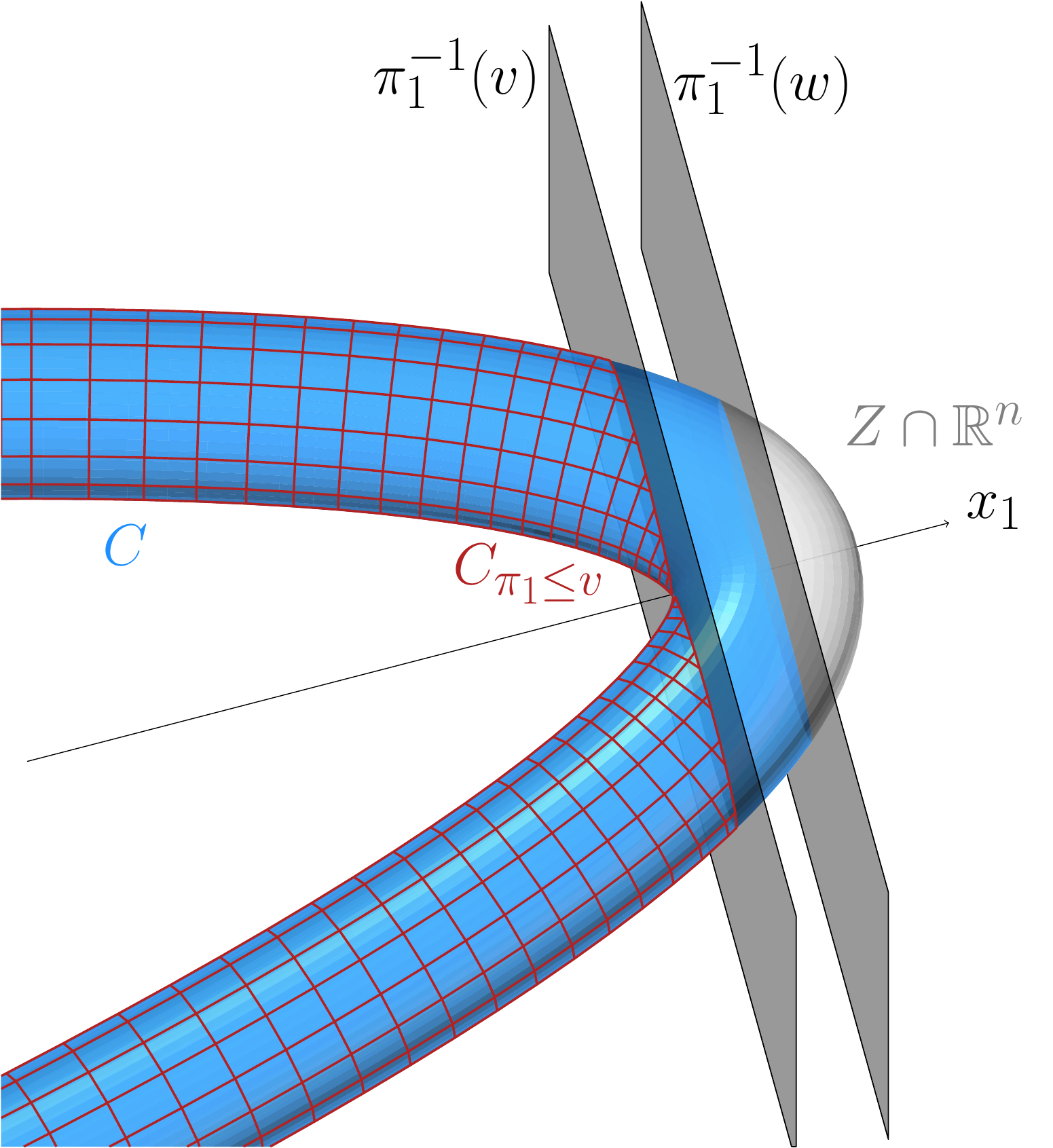}
\end{minipage}\hfill\
\begin{minipage}[c]{0.49\linewidth}
 \includegraphics[width=\linewidth]{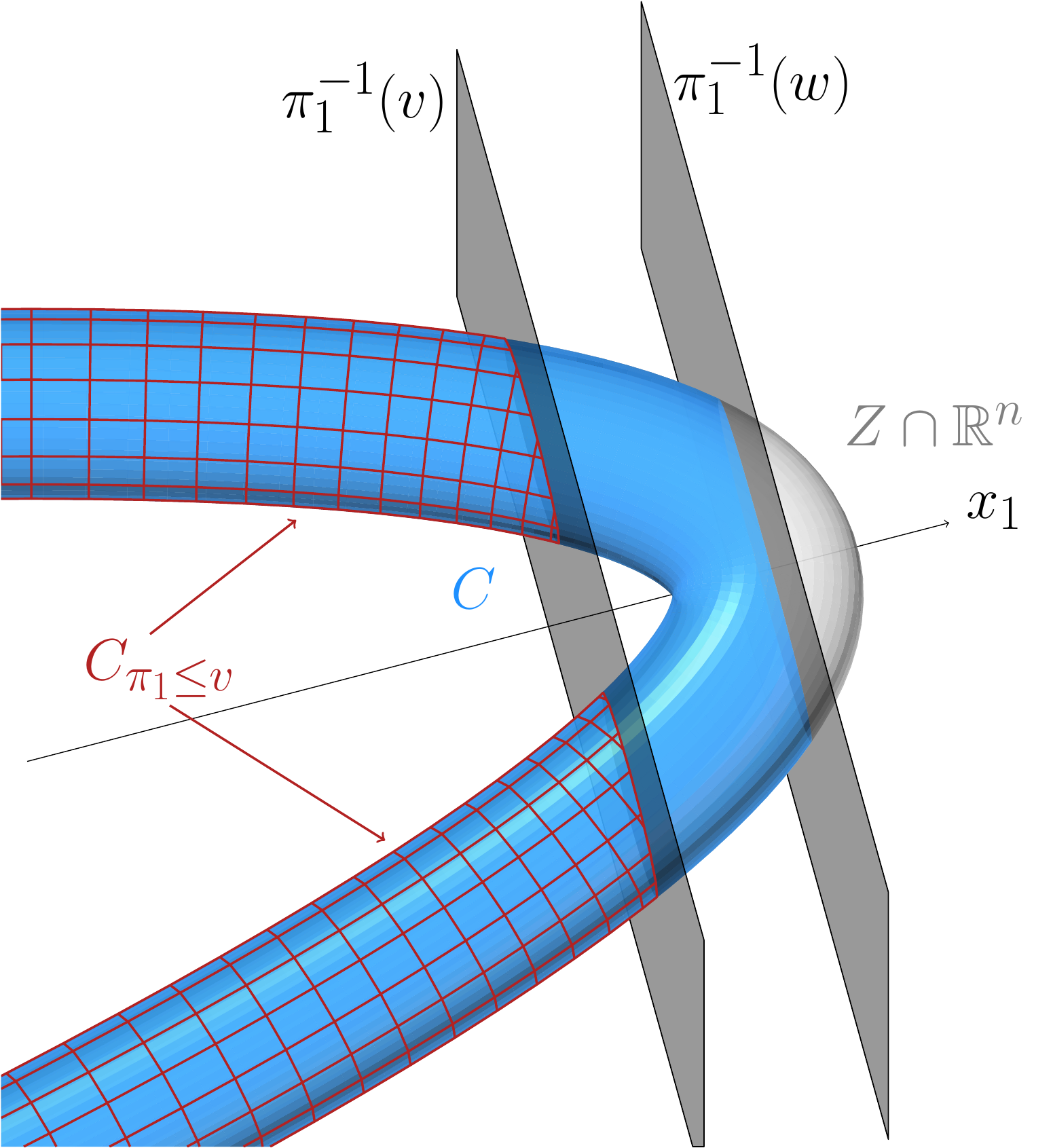}
\end{minipage}
\caption{Illustration of Corollary~\ref{cor:secondresult} where $\bphi = \pi_1$ 
and $Z$ is isomorphic to $\V(x_1^2+x_2^2-1)\times \rev{\V(x_1+x_2^2)}$ in two 
cases.
On the left $Z\rinto[\pi_1]{(v,w)} \cap K(\pi_1,Z) = \emptyset$ and we see that 
$C\rinfeq[\pi_1]{v}$ is still a \SACC of $Z\rinfeq[\pi_1]{v}$. On the 
right $Z\rinto[\pi_1]{(v,w)} \cap K(\pi_1,Z)$ contains a point and we see that 
$C\rinfeq[\pi_1]{v}$ is semi-algebraically disconnected.}
\label{fig:secondresultcor}
\end{figure}

 \section{Proof of the main connectivity result}\label{sec:proof}

Recall that $\bphi = (\phi_1, \ldots, \phi_n) \subset \RR[x_1,\dotsc,x_n]$ and 
for $1\leq i  \leq n, \map{\phi}{i} \colon \yy \mapsto (\phi_1(\yy), \ldots, 
\phi_i(\yy))$.
We denote by $W_i = W(\map{\phi}{i},V)$ the Zariski closure of the set of
critical points of the restriction of $\map{\phi}{i}$ to $V$ and recall that
\[
 K_i = W(\map{\phi}{1}, V) \cup \SW \cup \sing(V)
 \quad \et \quad F_i = \map{\phi}{i-1}[-1](\map{\phi}{i-1}(K_i))\cap V,
\]
where \rev{$\SW$ is a given subset of $V$}.
We suppose that the following assumptions hold:
\begin{description}
\item[$(\sfA)$] $V$ is $d$-equidimensional and its singular locus $\sing(V)$ is
  finite;
\item[$(\sfP)$] the restriction of the map $\map{\phi}{1}$ to $V\cap \RR^n$ is 
proper and bounded from below; 
\item[$(\sfB_1)$] $W_i$ is either empty or $(i-1)$-equidimensional and smooth 
outside $\sing(V)$;
\item[$(\sfB_2)$] for any $\yy = (\yy_1, \ldots, \yy_i)\in \CC^i$,
    $V\cap\map{\phi}{i-1}[-1](\yy)$ is either empty or 
    $(d-i+1)$-equidimensional;
\item[$(\sfC_1)$] $\SW$ is finite;
\item[$(\sfC_2)$] $\SW$ has a non-empty intersection with every \SACC of 
$W(\phiun,W_i)\cap\RR^n$.
\end{description}

Then the goal of this section is to prove that $W_i \cup F_i$ 
intersects each \SACC of $V \cap \RR^n$ and that their intersection is \SAC.

Let $\Rcal = F_i \cup W_i$. We prove that the following so-called roadmap 
property holds:
\begin{center}\itshape
\RM: ``For any \SACC $C$ of $V\cap\RR^n$, the set $C \cap \Rcal$ is non-empty 
and \SAC'',
\end{center}
by proving a truncated version of \RM and show that \rev{it} is enough.
For $u\in \RR$ let 
\begin{center}\itshape
\RMpu: ``\textit{For any \SACC $C$ of $V\rinfeq[ \map{\phi}{1}]{u}$, 
the set $C \cap \Rcal$ is non-empty and \SAC}''.
\end{center}
\begin{lemma}
 If \RMpu holds for all $u \in \RR$, then \RM holds. 
\end{lemma}
\begin{proof}
Let $C$ be a \SACC of $V \cap \RR^n$. 
Since $C$ is non-empty and \SAC, there exist $\yy$ and $\yy'$ 
in $C$, and a semi-algebraic path $\gamma\colon [0,1] 
\rightarrow C$ connecting them. Let 
\[
    u = \max\{ \map{\phi}{1}(\gamma(t)), t \in [0,1]\} \in \RR.
\]
Such a maximum $u$ exists by continuity of $\gamma$ and $ \map{\phi}{1}$, since
$[0,1]$ is closed and bounded, and it follows that $\gamma([0,1]) \subset 
V\rinfeq[ \map{\phi}{1}]{u}$.
Since $\gamma([0,1])$ is \SAC, there exists a (unique) \SACC $B$ of $V\rinfeq[ 
\map{\phi}{1}]{u}$ containing $\gamma([0,1])$. 
In particular, $B$ contains $\yy$ and $\yy'$. 
Since \RMpu[u] holds by assumption, then $B \cap \Rcal$ is non-empty. 
But as $\yy \in B \cap C$ and $B$ is \SAC,  $C$ contains $B$. 
Finally, $C\cap \Rcal$ contains $B\cap \Rcal$ and the former is non-empty.

We can suppose now, in addition, that $\yy$ and $\yy'$ are in $C \cap \Rcal$, 
and let $B$ be defined as above.
Then, $\yy$ and $\yy'$ are in $B \cap \Rcal$, which is \SAC by \RMpu[u]. 
Therefore $\yy$ and $\yy'$ are connected by a semi-algebraic path in 
$B\cap\Rcal$. Since $B \subset C$, $\yy$ and $\yy'$ are semi-algebraically 
connected in $C\cap \Rcal$. 
In conclusion, $C \cap \Rcal$ is \SAC and \RM holds.
\end{proof}

\begin{remark}
 The previous lemma trivially holds in the case of \cite[Theorem 14]{SS2011}, 
since $V \cap \RR^n$ is assumed to be bounded. 
Indeed, in this case, considering $u = \max_{\yy \in V \cap \RR^n}  
\map{\phi}{1}(\yy)$, one has $V\rinfeq[ \map{\phi}{1}]{u} = V\cap \RR^n$.
\end{remark}

\subsection{Restoring connectivity}
Before proving \RMpu for all $u \in \RR$, we need to prove the 
following result, which constitutes the core of the proof of 
Theorem~\ref{thm:mainresult}. 
This proposition shows that the connectivity property of our roadmap candidate 
is satisfied when $u$ is increasing towards singular points of $\map{\phi}{1}$ 
on $V$. This is ensured by the addition of the fibers $F_i$.

\begin{proposition}\label{prop:remainingstatement}
Let $u \in \RR$ and $C$ be a \SACC of $V\rinfeq[\map{\phi}{1}]{u}$ such 
that $C\rinf[ \map{\phi}{1}]{u}$ is non-empty. 
Let $B$ be a \SACC of $C\rinf[ \map{\phi}{1}]{u}$, then:
\begin{enumerate}
 \item $\bar{B} \cap (F_i \cup W_i)$ is non-empty;
 \item  Any point $\yy \in \bar{B} \cap (F_i \cup W_i)$ can be connected to 
a point $\zz \in B \cap (F_i \cup W_i)$ by a semi-algebraic path in $\bar{B} 
\cap (F_i \cup W_i)$.
\end{enumerate}
\end{proposition}

Let us begin with a technical lemma:

\begin{lemma}~\label{lem:fibrecontientptcrit}
Let $\KK$ be a real closed field containing $\RR$ and $\KKbar$ be its algebraic 
closure.
Let $Z\subset\KKbar^n$ be a $d$-equidimensional algebraic set, where $d>0$. 
Assume that for any $\zz \in \KKbar^{i-1}$, 
\begin{center}
$Z\cap\map{\phi}{i-1}[-1](\zz)$ is either empty or $(d-i+1)$-equidimensional.
\end{center}
Let $B$ be a \emph{bounded} \SACC of $Z\cap \KK^n$ and let $\yy \in B$.
Let $H$ be the \SACC of $B \cap \map{\phi}{i-1}[-1]( \map{\phi}{i-1}(\yy))$ 
containing $\yy$. Then, the intersection $H \cap K( \map{\phi}{i},Z)$ is not 
empty.
\end{lemma}

\begin{proof}
Let $Y = Z \cap \map{\phi}{i-1}[-1](\map{\phi}{i-1}(\yy))$. 
By assumption, $Y$ is an equidimensional algebraic set of dimension $d-i+1$.
Besides, $H$ is a bounded \SACC of $Y \cap \KK^n$, since $B$ is a 
bounded \SACC of $Z \cap \KK^n$.

Recall that $\bphi = (\phi_1,\dotsc,\phi_n)$.
Then $\phi_i(H) \subset \RR$ is a closed and bounded semi-algebraic set by 
\cite[Theorem 3.23]{BPR2006}. 
In particular, $\phi_i$ reaches its minimum on $H$.
Let $\zz \in H$ be such that $\phi_i(\zz)= \min \phi_i(H)$, so that 
$H\rinf[\phi_i]{\phi_i(\zz)}$ is empty. Then, by Lemma~\ref{lem:SACCvide},
\[
    \zz \in H \cap K(\phi_i,Y).
\]
Let $\gg \subset \KK[x_1,\dotsc,x_n]$ be a sequence of generators of $\I(Z)$,
so that $Y = \V(\gg, \map{\phi}{i-1} - \map{\phi}{i-1}(\yy))$.
Since $Y$ is $(d-i+1)$-equidimensional, Lemma~\ref{lem:caraccritrank} 
establishes that $\zz$ is such that 
\[
\rank \left[\begin{array}{c}\jac_\zz(\gg) \\ \jac_\zz( \map{\phi}{i-1}) \\ 
\jac_\zz(\phi_i) 
\end{array}\right] < n-(d-(i-1))+1.
\]
Since $ \map{\phi}{i} = ( \map{\phi}{i-1},\phi_i)$, one deduces that
\[
 \rank \left[\begin{array}{c}\jac_\zz(\gg) \\ \jac_\zz( \map{\phi}{i})
\end{array}\right] < n-d +i,
\]
which means that $\zz \in H \cap K( \map{\phi}{i},Z)$. 
Finally, the latter set is non-empty and the statement is proved.
\end{proof}

\begin{notation}
For the rest of the subsection let $u$, $C$ and $B$ as defined in
Proposition~\ref{prop:remainingstatement}.
\end{notation}
Let us deal with one particular case of the second item of 
Proposition~\ref{prop:remainingstatement}.
\begin{lemma}\label{lem:case2}
Let $\yy$ be in $\bar{B} \cap F_i$. Then, there exists a point 
$\zz \in B \cap (F_i \cup W_i)$ and a semi-algebraic path in 
$\bar{B} \cap (F_i \cup W_i)$ connecting $\yy$ to $\zz$. 
\end{lemma}
\begin{proof} 
Let $\yy$ be in $\in \bar{B} \cap F_i$. We assume that $\yy\notin B$ so that 
$ \map{\phi}{1}(\yy) = u$, otherwise taking $\zz = \yy$ would end the proof.
Since $\yy \in \bar{B}$, by the curve selection lemma \cite[Th. 3.22]{BPR2006}, 
there exists a semi-algebraic path $\gamma\colon [0,1] \rightarrow \RR^n$ such 
that $\gamma(0) = \yy$ and $\gamma(t) \in B$ for all $t \in (0,1]$. 
Let $\eps$ be an infinitesimal, $\RR' = \RR\left\langle\eps\right\rangle$ be 
the field of algebraic Puiseux series and $\psi = (\psi_1,\dotsc,\psi_n)$ be 
the semi-algebraic germ of $\gamma$ at the right of the origin (see 
\cite[Section 3.3]{BPR2006}). 
According to \cite[Theorem 3.17]{BPR2006}, we can identify $\psi$ 
with an element of $(\RR')^n$ (by a slight abuse of notation, we will denote 
them 
in the same manner). 
Hence by \cite[Proposition 3.21]{BPR2006}, $\limeps{\psi} = \yy$.
Let finally
\[
H = \ext(B,\RR') \cap  \map{\phi}{i-1}[-1]( \map{\phi}{i-1}(\psi)) \subset 
(\RR')^n
\]
where $\ext(B,\RR')$ is the extension of $B$ to $\RR'$ and
$\map{\phi}{j}$ for $1\leq j \leq n$, with some notation abuse, still denote 
the extension of $\map{\phi}{j}$ to $\RR'$. 

Since $\gamma((0,1)) \subset B$, by \cite[Proposition 3.19]{BPR2006},
$\psi$ is  in $\ext(B,\RR')$. Hence, $\psi$ in $H$ and $H$ is non-empty. 
Moreover $B$ is bounded since $\phiun\colon V \cap \RR^n\to\RR$ is a proper map 
bounded below by assumption $(\sfP)$. Then \cite[Proposition 3.19]{BPR2006} 
states that $\ext(B,\RR')$ and then $H$ are bounded over $\RR$. Hence the 
map $\limeps$ is well defined on $H$ and
\[
	\yy \in \limeps{H} = \{ \limeps{\yy'}, \: \yy' \in H \} \subset \RR^n.
\]
Finally, as $\map{\phi}{i-1}$ is semi-algebraic and continuous, $\limeps{H}$ is 
contained in $\bar{B} \cap \map{\phi}{i-1}[-1](\map{\phi}{i-1}(\yy))$ by 
\cite[Lemma 3.24]{BPR2006}. But $\yy \in F_i$, so that 
\[
  \map{\phi}{i-1}[-1](\map{\phi}{i-1}(\yy)) \subset 
 \map{\phi}{i-1}[-1](\map{\phi}{i-1}(K_i)),
\]
and finally $\limeps H$ is actually in $\bar{B}\cap F_i$.

Let $H_1$ be  the \SACC of $H$ containing $\psi$. By \cite[Proposition 
5.24]{BPR2006}, $\limeps H_1$ is the \SACC of $\limeps H$ containing $\yy$.
Actually, we just proved that every $\ww$ in $\limeps H_1$ can be 
semi-algebraically connected to $\yy$ into $\bar{B} \cap F_i$. 
We find now some $\ww\in\limeps H_1$ that can be connected to a point $\zz \in 
B 
\cap (F_i \cup W_i)$ to end the proof. Such a $\ww$ must be the origin of a 
germ of semi-algebraic functions that lies in $B \cap (W_i \cup F_i)$.

By assumption $(\sfA)$, $V$ is $d$-equidimensional.
By assumption $(\sfB_2)$, for all $\zz \in V$, the algebraic set $V \cap 
\map{\phi}{i-1}[-1]( \map{\phi}{i-1}(\zz))$ is $(d-i+1)$-equidimensional.
Then, if we denote by $\CC'$ the algebraic closure of $\RR'$, it is an 
algebraic closed extension of $\CC$, so that the algebraic sets of $(\CC')^n$ 
\[
   Z = \big\{ \zz \in (\CC')^n \mid \forall h \in \I(V), h(\zz) = 0 \big\} \et
   Z \cap \map{\phi}{i-1}[-1]( \map{\phi}{i-1}(\psi)))
\]
are equidimensional of dimension respectively $d$ and $(d-i+1)$.
Since $B$ is a \SACC of $V\rinf[ \map{\phi}{1}]{u}$, then, by 
\cite[Proposition 5.24]{BPR2006}, $\ext(B,\RR')$ is a \SACC of
\[
 \ext(V\rinf[ \map{\phi}{1}]{u},\RR') =
 \ext(V\cap\RR^n,\RR')\rinf[\map{\phi}{1}]{u}
 = Z\rinf[\phiun]{u},
\]
by \cite[Transfer Principle, Th. 2.98]{BPR2006}.
Then, since $H_1$ is a \SACC of $H=\ext(B,\RR')\cap\map{\phi}{i-1}[-1]( 
\map{\phi}{i-1}(\psi))$, one can apply Lemma~\ref{lem:fibrecontientptcrit}
on $Z$ with $\KK = \RR'$. Hence
\[
     H_1 \cap K( \map{\phi}{i}, Z) \neq \emptyset.
\]
By Lemma~\ref{lem:caraccritminor}, $K(\map{\phi}{i},Z)$ is defined over $\RR$ 
as $V$  and $\map{\phi}{i}$ are. 
Then, by \cite[Transfer Principle, Th. 2.98]{BPR2006},
\[
 K( \map{\phi}{i},Z) \cap (\RR')^n = \ext(K( 
\map{\phi}{i},V)\cap \RR^n,\RR'),
\]
so that 
\[
 \emptyset \quad \subsetneq 
 \quad H_1 \cap \ext(K( \map{\phi}{i},V) \cap\RR^n,\RR') \quad \subset \quad
 \ext(B \cap K( \map{\phi}{i},V),\RR').
\]
Therefore let $\zeta \in \ext(B \cap K( \map{\phi}{i},V),\RR')$, let $\ww = 
\limeps{\zeta}$ and $\tau$ be a representative of $\zeta$ on $(0,t_0)$, where 
$t_0 >0$. 
By \cite[Proposition 3.21]{BPR2006}, we can continuously extend $\tau$ to 0 
such that $\tau(0) = \ww$. 
Besides for all $t\in(0,t_0)$, 
\[
 \tau(t) \in B \cap K( \map{\phi}{i},V) \subset B \cap (W_i \cup F_i).
\]
Then $\tau([0,t_0)) \subset \bar{B} \cap (F_i \cup W_i)$ so that
\[
\ww \in \bar{B} \cap (F_i \cup W_i) \et \zz = \tau(t_0/2) \in B \cap 
(F_i \cup W_i).
\]
Besides, since $\ww \in \limeps H_1$ we have seen that it can be connected to 
$\yy$ a semi-algebraic path in $\bar{B} \cap (F_i \cup W_i)$.
In the end, there exist two consecutive paths into $\bar{B}\cap(F_i \cup W_i)$, 
connecting $\yy$ to $\ww$, and $\ww$ to $\zz \in B \cap \Rcal$ (namely $\tau$).
\end{proof}

\begin{figure}[h]\centering
  \includegraphics[width=\linewidth]{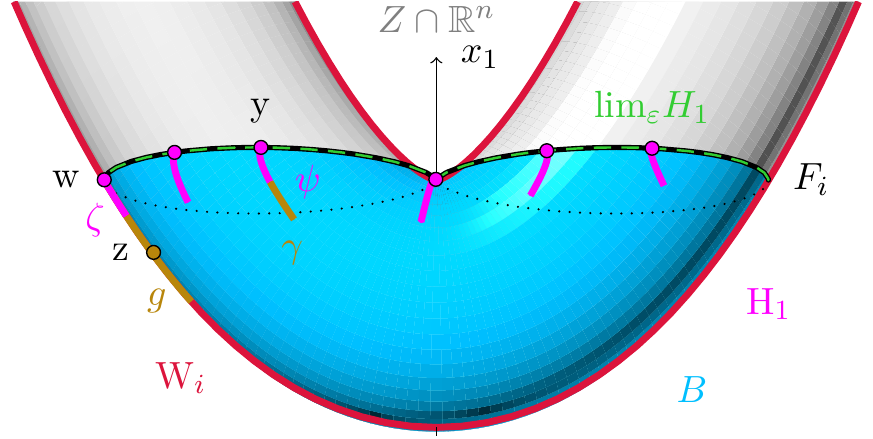}
  \caption{Illustration of proof of Lemma~\ref{lem:case2} with
  $\map{\phi}{1} = \pi_1$ and $V$ is isomorphic to 
 $\V(x_1^2+x_2^2-1)\times \rev{\V(x_1-x_2^2)}$. 
  Elements of $H_1$ can be seen as curves of infinitesimal lengths, starting 
from a point of $\limeps H_1$, and lying in $B$.
  Here, $\limeps H_1$ is the set of points that share the same first coordinate 
  than $\yy$. Hence, the above proof consisted in choosing a $\zeta$ in $H_1$, 
  that lives ``inside'' $W_i\cup\sing(V)$ (actually in 
$\ext(W_i\cup\sing(V),\RPuis)$).}
\label{fig:remainingstate}
\end{figure}

We can now prove Proposition~\ref{prop:remainingstatement}. This proof is 
illustrated by Figure~\ref{fig:remainingstate}.
\begin{proof}[Proof of Proposition~\ref{prop:remainingstatement}]
Let $B$ be a \SACC of $C\rinf[ \map{\phi}{1}]{u}$. Since $ \map{\phi}{1}$ is a 
proper map bounded from below on $V\cap \RR^n$ by assumption $\sfP$, 
$C\rinf[\map{\phi}{1}]{u}$, and then $B$, are bounded.
Then applying Proposition~\ref{prop:SACCcontientptcrit} shows that:
\[
 \emptyset \; \subsetneq \; B \cap K( \map{\phi}{1},V) \; \subset \;  B\cap 
F_i \; \subset \; B \cap (F_i \cup W_i).
\]
The first item is then proved.
Let $\yy \in \bar{B} \cap(F_i \cup W_i)$.
To prove the second item, one only needs to consider the case where $\yy \in  
\bar{B} \cap (W_i - F_i)$ according to Lemma~\ref{lem:case2}.
Moreover one can assume that $\yy \notin B$ and then $\map{\phi}{1}(\yy) = u$, 
otherwise, taking $\zz = \yy$, would end the proof. 

Let $D$ be the \SACC of $(W_i)\rinfeq[\map{\phi}{1}]{u}$ containing $\yy$.
\rev{We consider} two disjoint cases.
\begin{enumerate}
 \item \emph{If $D \not\subset \bar{B}$,}
 there exists $\yy' \in D$ such that $\yy' \notin \bar{B}$. 
 Then let $\gamma\colon [0,1] \rightarrow D$ such that $\gamma(0)=\yy$ and 
$\gamma(1)=\yy'$.
 Hence, if 
 \[
 t_1 = \max\{t\in[0,1) \mid \gamma(t) \in \bar{B}\},
 \]
 then $\gamma(t_1) \in K(\map{\phi}{1},V)$ by the contrapositive of statement 
$c)$ 
of Lemma~\ref{lem:voisinageSAC}.
 Since $K(\map{\phi}{1},V)\subset F_i$, we can apply 
Lemma~\ref{lem:case2} to $\gamma(t_1)$ and find $\zz \in B\cap(F_i\cup W_i)$ 
that is connected to $\gamma(t_1)$ and then to $\yy$ by a semi-algebraic path 
in 
$\bar{B} \cap (F_i \cup W_i)$.
 
 \item \emph{If $\mathit{D \subset \bar{B}}$,} we claim that there exists some 
$\zz \in D\cap F_i$. 
 Indeed since $D$ is a \SACC of $(W_i)\rinfeq[\phiun]{u}$ and $\phiun$ is a 
proper map, $D$ is bounded. 
 Then by Proposition~\ref{prop:SACCcontientptcrit} there exists $\yy' \in D 
\cap K(\phiun,W_i)$. 
 If $\yy' \in \sing(W_i)$ then $\yy' \in \sing(V)$ by assumption $\sfB_1$ and 
taking $\zz = \yy' \in F_i$ one concludes as in the first item. 

 Else $\yy'$  is in $W(\phiun,W_i)$, and we let $E$ be the \SACC of 
$W(\phiun,W_i)$ 
containing $\yy'$. 
 Since $\phiun(W(\phiun,W_i))$ is finite by Sard's lemma, $\phiun(E) = 
\{\phiun(\yy')\}$, so that $E \subset (W_i)\rinfeq[\phiun]{u}$.
 Hence, since $E$ is \SAC, $E \subset D$.
 By assumption $\sfC_2$, there exists $\zz \in E \cap \SW$,
 so that $\zz \in D \cap \SW \subset D \cap F_i$ and we are done.
 
 Then we can connect $\yy$ to $\zz$ inside $D \subset \bar{B} \cap W_i$ and 
since $\zz$ is in $D \cap F_i$, which is contained in $\bar{B} \cap F_i$, we 
can 
connect similarly 
$\zz$ to some $\zz'\in B \cap (F_i\cup W_i)$ inside $\bar{B} \cap F_i$ by 
Lemma~\ref{lem:case2}.
Putting things together, $\yy$ is connected to some $\zz'\in B \cap (F_i\cup 
W_i)$ by a semi-algebraic path in $\bar{B} \cap F_i$.\\[-3em]
\end{enumerate}
\end{proof}

\begin{corollary}\label{cor:remainingstatement}
 Let $u \in \RR$ such that for all $u'<u$, \RMpu[u'] holds. Let $C$ be a 
\SACC of $V\rinfeq[ \map{\phi}{1}]{u}$ such that $C\rinf[\map{\phi}{1}]{u}$ is 
non-empty. 
If $B$ is a \SACC of $C\rinf[ \map{\phi}{1}]{u}$, then $\bar{B}\cap\Rcal$ is 
non-empty and \SAC.
\end{corollary}
\begin{proof}
Let $\yy$ and $\yy'$ be in $\bar{B} \cap \Rcal$. According to 
Proposition~\ref{prop:remainingstatement}, they can respectively be
connected to some $\zz$ and $\zz'$ in $B \cap \Rcal$, by a semi-algebraic path 
in $\bar{B} \cap \Rcal$. 
As $B$ is \SAC, there exists a semi-algebraic path $\gamma\colon [0,1] 
\rightarrow B$ connecting $\zz$ to $\zz'$. Let 
\[
u'=\max \big\{  \map{\phi}{1}(\gamma(t)) \mid t \in [0,1]\big\}, 
\]
so that $\gamma([0,1]) \subset V\rinfeq[ \map{\phi}{1}]{u'}$. Such a $u'$ 
exists by continuity of $\gamma$, and satisfies $u'<u$, as $[0,1]$ is closed 
and bounded.

Let $B'$ be the \SACC of $B\rinfeq[ \map{\phi}{1}]{u'}$ that contains 
$\gamma([0,1])$. Since $B'$ is also a \SACC of $V\rinfeq[ \map{\phi}{1}]{u'}$, 
property \RMpu[u'] states that $B' \cap \Rcal$ is non-empty and \SAC. 
Then, as $\zz$ and $\zz'$ are in $B' \cap \Rcal$, they can be connected by a 
semi-algebraic path in $B' \cap \Rcal$, and then, in $B \cap \Rcal$. 
Thus $\yy$ and $\yy'$ are connected by a semi-algebraic path in $\bar{B} 
\cap \Rcal$ and we are done.
\end{proof}

\subsection{Recursive proof of the truncated roadmap property}
In order to prove that \RMpu holds for all $u \in \RR$, one can consider two 
disjoint cases: whether $u$ is a real singular value of $\map{\phi}{1}$,
that is $u \in \map{\phi}{1}(K_i)$, or not.
The following lemma allows us to proceed by induction.

\begin{lemma}
 The set $ \map{\phi}{1}(K_i)$ is non-empty and finite.
\end{lemma}
\begin{proof}
By the algebraic version of Sard's theorem \cite[Proposition 
B.2]{SS2017}, the set of critical values of $\map{\phi}{1}$ on $V$ is an 
algebraic set of $\CC$ of dimension 0. Then, it is either empty or non-empty 
but finite. 
Hence, $\map{\phi}{1}(K_i)$ is either empty or non-empty but finite, as 
$\SW$ and $\sing(V)$ are, by assumption.
Moreover since $ \map{\phi}{1}$ is a proper map bounded from below on $V \cap 
\RR^n$ by assumption $(\sfP)$, for any $u\in \RR$, $Z\rinf{u}$ is bounded. 
Then, 
since $V$ is not empty, by Proposition~\ref{prop:SACCcontientptcrit} the sets 
$K(\map{\phi}{1},V)$ and then $ \map{\phi}{1}(K_i)$ are not empty.
\end{proof}

We denote by $v_1 < \dotsc < v_\ell$ the points of $\map{\phi}{1}(K_i\cap 
\RR^n)$ and, in addition, let $v_{\ell+1} = +\infty$. 
We proceed by proving the two following 
steps.
\begin{description}
    \item[Step 1:] Let $u \in \RR$, if \RMpu[u'] holds for all $u'< 
    u$, then \RMpu holds.
    \item[Step 2:] Let $j \in \{1,\dotsc, \ell\}$, if \RMpu[v_j] holds, 
    then for all $u \in (v_j,v_{j+1})$, \RMpu holds.
\end{description}
Remark that, by Lemma~\ref{lem:SACCvide}, $v_1 = \min_{V\cap\RR^n} 
\map{\phi}{1}$, since $V\cap\RR^n$ is closed. 
Then for $u'<v_1$, $V\rinfeq{u'}=\emptyset$ and \RMpu[u'] trivially holds.
Hence, proving these two steps is enough to prove \RMpu for all $u$ 
in $\RR$, by an immediate induction.

\begin{proposition}[\textbf{Step 1}]\label{prop:step1}
 Let $u \in \RR$. Assume that for all $u'<u$, \RMpu[u'] holds. Then \RMpu holds.
\end{proposition}
The proof of this proposition is illustrated by Figure~\ref{fig:step1}.
\begin{proof}
Let $u \in \RR$ be such that for all $u'<u$, \RMpu[u'] holds and let $C$ be 
a \SACC of $V\rinfeq[ \map{\phi}{1}]{u}$. 
We have to prove that $C\cap \Rcal$ is non-empty and \SAC.

If $C\rinf[ \map{\phi}{1}]{u}$ is empty, then, by Lemma~\ref{lem:SACCvide}, 
$C \subset K( \map{\phi}{1},V)$.
But the points of $K( \map{\phi}{1},V)$ are either in $W_i$ or in 
$\sing(V) \subset F_i$.
Hence $K(\map{\phi}{1},V) \subset \Rcal$ and $C\cap\Rcal = C$, which is 
non-empty and \SAC by definition.

From now on, $C\rinf[ \map{\phi}{1}]{u}$ is supposed to be non-empty and let 
$B_1,\dotsc,B_r$ be its \SAC components. 
According to Corollary~\ref{cor:remainingstatement}, for all $1 \leq j \leq r$, 
$\overline{B_j} \cap \Rcal$ is non-empty and \SAC.
Then, as $\overline{B_j} \subset C$, 
\[
 \overline{B_j}\cap \Rcal \subset C \cap \Rcal 
\]
for every $1 \leq j \leq r$, and $C \cap \Rcal$ is non-empty.

Let us now prove that $C \cap \Rcal$ is \SAC. Let $\yy$ and $\yy'$ in $C 
\cap \Rcal$. As $C$ is \SAC, there exists a semi-algebraically continuous map 
$\gamma\colon [0,1] \rightarrow C$ such that $\gamma(0)=\yy$ and 
$\gamma(1)=\yy'$. Now let 
\[
 G = \gamma^{-1}(C\req[ \map{\phi}{1}]{u} \cap K( \map{\phi}{1},V)) \et H = 
[0,1] - G.
 \]
 We denote by $G_1,\dotsc, G_N$ the connected components of $G$ and 
$H_1,\dotsc,H_M$ those of $H$.
The sets $H_j$ for $1 \leq j \leq M$ are open intervals of $[0,1]$, and we note 
$\ell_j = \inf(H_j)$ and $r_j = \sup(H_j)$. 
Since $\gamma(G)$ already lies in $C \cap \Rcal$, let us establish that for 
every $1 \leq j \leq M$, $\gamma(\ell_j)$ and $\gamma(r_j)$ can be connected by 
another semi-algebraic path $\tau_j$ in $C \cap \Rcal$.
 
Let $1\leq j \leq M$, then $\gamma(H_j) \cap (C\req[ \map{\phi}{1}]{u} \cap 
K( \map{\phi}{1},V)) = \emptyset$ by definition. 
Moreover, $\gamma(H_j) \subset C$ so that
\[
 \gamma(H_j) \cap (V\req[ \map{\phi}{1}]{u} \cap K( \map{\phi}{1},V)) = 
\emptyset.
\] 
Hence, since $H_j$ is connected, there exists (by 
Proposition~\ref{prop:firstresult}) a unique \SACC $B$ of $V\rinf[ 
\map{\phi}{1}]{u}$ such that $\gamma(H_j) \subset \bar{B}$.
But $\gamma(H_j) \subset C$, so that $\bar{B}$ and thus $B$ are actually 
contained in $C$. Therefore, $B$ is actually a \SACC of $C\rinf[ 
\map{\phi}{1}]{u}$ and there exists $1 \leq k \leq r$ such that $B = B_k$.
At this step $\gamma(H_j) \subset \bar{B_k}$, so that
\[
\gamma([\ell_j,r_j]) = \gamma(\bar{H_j}) \; \subset \; \bar{\gamma(H_j)} 
\; \subset \; \bar{B_k},
\]
and both $\gamma(\ell_j)$ and $\gamma(r_j)$ are in $\bar{B_k}$.
Remark that both $\ell_j$ and $r_j$ are in $G$, so that both $\gamma(\ell_j)$ 
and $\gamma(r_j)$ are in $K(\map{\phi}{1},V) \subset F_i \subset 
\Rcal$.
Thus, both $\gamma(\ell_j)$ and $\gamma(r_j)$ are in $\bar{B_k} \cap \Rcal$. 
According to Corollary~\ref{cor:remainingstatement}, they can be connected by a 
semi-algebraic path $\tau_j\colon [0,1] \to \bar{B_k} \cap \Rcal \subset C \cap 
\Rcal$.
 
In conclusion, we have proved that for $1 \leq j \leq M$, $\gamma(\ell_j)$ and 
$\gamma(r_j)$ can be connected by a semi-algebraic path $\tau_j$ in $C \cap 
\Rcal$.
Therefore the semi-algebraic sub-paths $\gamma_{\mid H_j}$ can be replaced by 
the $\tau_j$'s, which lie in $C \cap \Rcal$. 
Moreover, for all $1 \leq j \leq N$
\[
\gamma(G_j) \subset C\cap \Rcal.
\]
Since the $H_j$'s and $G_j$'s form a partition of $[0,1]$, by putting together 
alternatively the $\tau_j$'s and the $\gamma_{\mid G_j}$'s, one obtains a 
semi-algebraic path in $C \cap \Rcal$ connecting $\yy=\gamma(0)$ to $\yy' 
=\gamma(1)$. And we are done.
\end{proof}

\begin{figure}[h]\centering
  \includegraphics[width=0.8\linewidth]{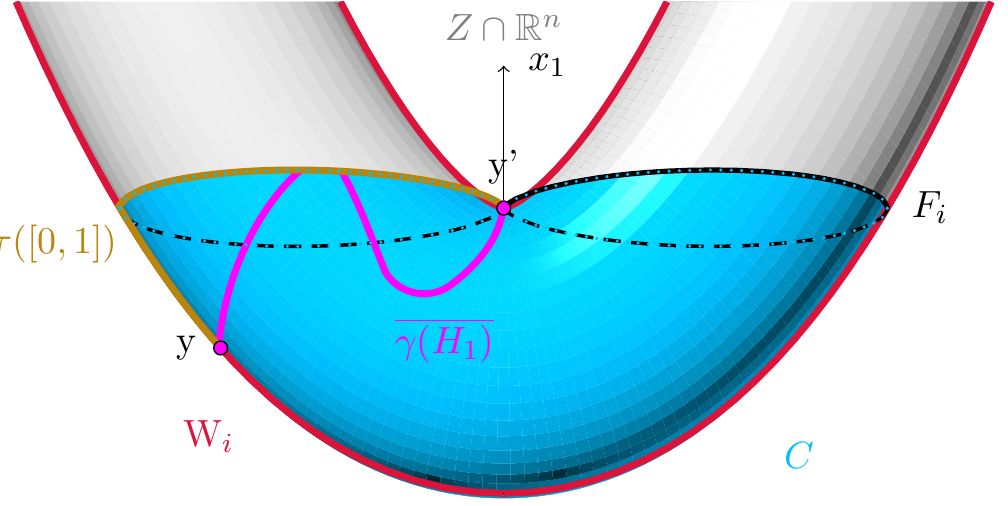}
  \caption{Illustration of proof of Proposition~\ref{prop:step1} with
$\map{\phi}{1} = \pi_1$ and $V$ is isomorphic to 
$\V(x_1^2+x_2^2-1)\times \rev{\V(x_1-x_2^2)}$.
  Here, only $\yy'$ belongs to $C\req[\pi_1]{u} \cap K(\pi_1,V)$. Then we 
replace the 
path $\gamma=\gamma_{\mid H_1}$ by a path $\tau_1$ that lies in the 
intersection of the roadmap and the \SACC $C$.}
 \label{fig:step1}
\end{figure}

\begin{proposition}[\textbf{Step 2}]\label{prop:step2}
    Let $j \in \{1,\dotsc,\ell\}$, if \RMpu[v_j] holds, then 
    for all $u  \in (v_j,v_{j+1})$, \RMpu holds.
\end{proposition}
The proof of this proposition is illustrated by Figure~\ref{fig:step2}.
\begin{proof}
Let $j \in \{0,\dotsc,\ell\}$ and $u \in (v_j,v_{j+1})$. Let $C$ be a \SACC of 
$V\rinfeq[ \map{\phi}{1}]{u}$; we have to prove that $C \cap \Rcal$ is 
non-empty and \SAC.
 
Let us first prove that $C\rinfeq[ \map{\phi}{1}]{v_j}\cap \Rcal$ is non-empty 
and \SAC. By assumption $(\sfA)$, $V$ is an equidimensional algebraic set of 
positive dimension, and by assumption $(\sfP)$, the restriction of 
$\map{\phi}{1}$ to $V\cap\RR^n$ is a proper map bounded below. 
Moreover, as $ \map{\phi}{1}\left(K( \map{\phi}{1},V) \cap \RR^n\right) \subset 
\{v_1,\dotsc,v_\ell\}$, then 
\[
V\rinto[ \map{\phi}{1}]{(v_j,u]} \cap K( \map{\phi}{1},V) = \emptyset.
\]
Then using Corollary~\ref{cor:secondresult}, one deduces that $C\rinfeq[ 
\map{\phi}{1}]{v_j}$ is a \SACC of $V\rinfeq[\map{\phi}{1}]{v_j}$. 
Hence, by property \RMpu[v_j], the set $C\rinfeq[ \map{\phi}{1}]{v_j} \cap 
\Rcal$ is non-empty and \SAC. In particular, $C \cap \Rcal$ is non-empty.
 
Let us now prove that $C \cap \Rcal$ is \SAC. Let $\yy$ be in $C \cap \Rcal$. 
According to the previous paragraph, one just need to be able to connect $\yy$ 
to a point $\zz$ of $C\rinfeq[ \map{\phi}{1}]{v_j} \cap \Rcal$ by a 
semi-algebraic path in $C \cap \Rcal$ and then apply \RMpu[v_j].
First, if $\yy \in C\rinfeq[ \map{\phi}{1}]{v_j} \cap \Rcal$, there is nothing 
to do. Suppose now that $\yy \in C\rinto[ \map{\phi}{1}]{(v_j,u]} \cap \Rcal$. 
We claim that actually
\[
\yy \in C \cap W_i.
\]
Indeed, if $\yy \in C \cap F_i$, then $ \map{\phi}{i-1}(\yy) \in 
 \map{\phi}{i-1}(K_i)$ and $ \map{\phi}{1}(\yy)$ would be one of the 
$v_1,\dotsc,v_\ell$. 

Let $D$ be the \SACC of $(C\cap W_i)\rinfeq[\phiun]{u}$ containing $\yy$. 
Remark that $D$ is a \SACC of $(W_i)\rinfeq[\phiun]{u}$, as it contains $\yy$ 
and is contained in $C$. Since $\phiun(W(\phiun,W_i))$ is finite by Sard's 
lemma, we get that \rev{$\phiun(W(\phiun,W_i)) \subset \phiun(\SW)$, by 
assumption $(\sfC_2)$, so that}
\[
(v_j,u) \cap \phiun(W(\phiun,W_i)) = \emptyset.
\]
Since $W_i$ is equidimensional and smooth outside $\sing(V)$, then by 
Corollary~\ref{cor:secondresult}, $D\rinfeq[\phiun]{v_j}$ is a \SACC of 
$(W_i)\rinfeq[\phiun]{v_j}$.
Therefore, let $\zz \in D\rinfeq[\map{\phi}{1}]{v_j}$. Since $D$ is \SAC, there 
exists a semi-algebraic path, connecting $\yy \in D \subset C \cap \Rcal$ to 
$$\zz \in D\rinfeq[ \map{\phi}{1}]{v_j} \subset C\rinfeq[ \map{\phi}{1}]{v_j} 
\cap \Rcal$$ in $D \subset C \cap \Rcal$.  
We are done.
\end{proof}

\begin{figure}[h]
 \begin{center}\centering
  \includegraphics[width=0.8\linewidth]{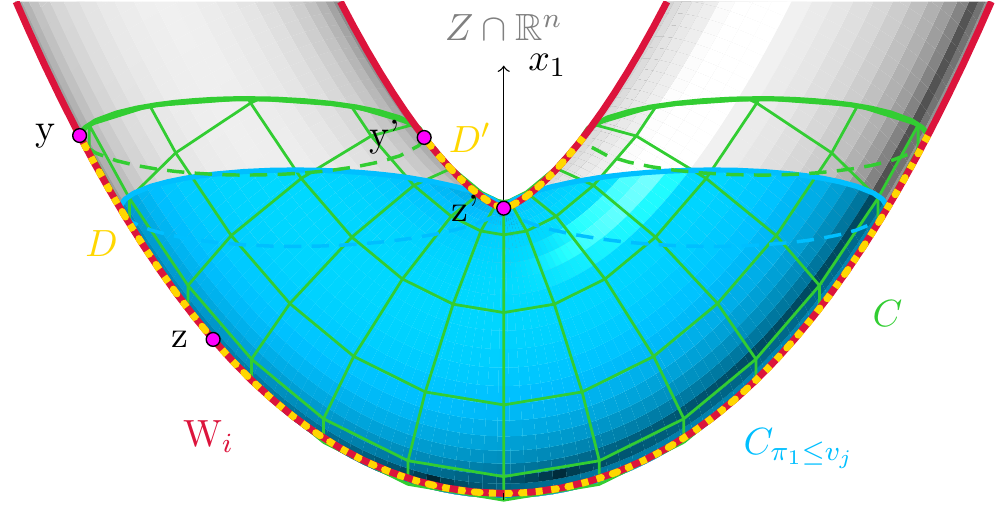}
  \caption{Illustration of proof of Proposition~\ref{prop:step2} with
  $\map{\phi}{1} = \pi_1$ and $V$ is isomorphic to 
$\V(x_1^2+x_2^2-1)\times \rev{\V(x_1-x_2^2)}$. 
  We connect the points $\yy$ and $\yy'$ in $C\cap W_i$ to respectively
  $\zz$ and $\zz'$ in $C\rinfeq[\pi_1]{v_j}$. Then we are reduced to 
  the case of Step 1.}
  \label{fig:step2}
 \end{center}
\end{figure}

 \section{Conclusions and perspectives}\label{sec:perspectives}

\rev{We illustrate below two ways of using Theorem~\ref{thm:mainresult} in order
  to generalize} the algorithms of \cite{SS2017} to the case of unbounded smooth
real algebraic sets.

Let $V\subset \CC^{n}$ be an equidimensional algebraic set of dimension $d$
\rev{given as the solutions of some polynomials $f_1, \ldots, f_p$ in
  $\QQ[x_1, \ldots, x_n]$}. Assume that $\sing(V)$ is finite. Take
\[
  \phi_1=\sum_{k=1}^n (x_k-\aa_k)^2 
  \quad \text{and for\: $2 \leq j \leq n$}\quad
  \phi_j = \sum_{k=1}^n \bb_{j,k}x_k, \qquad 
\] 
where $\aa = (\aa_1,\dotsc,\aa_n) \in \QQ^n$ and, for $2 \leq j \leq n$, $\bb_j
= (\bb_{j,1},\dotsc, \bb_{j,n}) \in \QQ^n$. Then, assumption $(\sfP)$ holds. 
Also,
according to \cite{BGHP2005, BGHSS2010}, for a generic choice of $\aa$ and 
$\bb$,
the dimension properties of assumption $(\sfB)$ do hold.

\rev{For some chosen $2\leq i \leq d$, let $W_i$ and $F_i$ be respectively the 
polar
  variety and set of fibers as defined in the introduction.} \rev{One can
  compute} a set $S\subset W(\phi_1, W_i)\subset V$ by using any algorithm such
as \cite[Chap. 13]{BPR2006} or~\cite{SaSc03}, returning sample points in all
connected components of real algebraic sets.

Hence, one can apply Theorem~\ref{thm:mainresult} to $V$, $\bphi$
\rev{and $i$}.
\rev{We deduce that $W_i\cup F_i$ has a non-empty and connected intersection
  with all connected components of $V\cap\RR^{n}$, but it is in general an
  object of dimension greater than $1$, so more work is needed.}

\rev{Our first option is to recursively perform similar operations,
  using this time polynomials defining respectively $W_i$ and $F_i$,
  eventually building a roadmap for~$V$ itself (this is justified by
  \cite[Prop. 2]{SS2011}). This requires that both $W_i$ and $F_i$ are
  equidimensional with finitely many singular points.}

\rev{The cost of the whole procedure will depend on the degrees of
  $W_i$ and $F_i$.  Denote by $D$ the degrees of the $f_i$'s and by
  $\delta$ (resp. $\sigma$) the degree of $V$ (resp. $S$); using
  \cite[Chap. 13]{BPR2006} gives $\sigma \in D^{O(n)}$. Assuming that
  $d = n-p$ and that the ideal generated by $f_1, \ldots, f_p$ is
  radical, one can apply Heintz's version of the Bézout theorem
  \cite{Heintz83} as in \cite{SaSc03} to deduce that $\delta \leq
  D^{p}$ and that the degree of $W_i$ is bounded by $\delta (nD)^{n-p}
  \in (nD)^{O(n)}$.  Similarly, one can expect that the degree
  $\delta'$ of $W(\phi_1, W_i)$ lies in $(nD)^{O(n)}$ by applying
  similar arguments to those used in \cite{SS2017}. Since the degree
  of $F_i$ is bounded by the product of $\delta$ and $\delta'+
  \sigma$, we deduce that its degree also lies in $(nD)^{O(n)}$.}

\rev{Hence, as explained in \cite{SS2011}, the overall complexity of such
  recursive algorithms is  $(nD)^{O(nr)}$, where $r$ is the depth of
  the recursion, provided that the involved geometric sets do satisfy the
  properties needed by Theorem~\ref{thm:mainresult} and can be represented and
  computed with algebraic data within complexities which are polynomial in their
  degrees.}

\rev{To understand the possible depth of recursion one could expect,
  one also needs to have a look at the dimensions of $W_i$ and
  $F_i$. Observe that $W_i$ is expected to have dimension $i -
  1$. Similarly, $F_i$ is expected to have dimension $d -
  (i-1)$. Taking $i\simeq \lfloor \frac{\dim(V)}{2} \rfloor$ will
  decrease the dimensions of $W_i$ and $F_i$ to $\simeq \lfloor
  \frac{\dim(V)}{2} \rfloor$ if they are not empty (this will require
  coordinates).  Hence the depth $r$ of this new recursive roadmap
  algorithm will be bounded by $\log_2(n)$.}

\rev{A second approach to design our new algorithm takes $i =
  2$. Then, $W_2$ is expected to have dimension $1$ (or be empty), so
  no further computation is needed. On the other hand, $F_2$ still has
  dimension $d - 1$, but a key observation is that $F_2$ is now {\em
    bounded}. Then, one can directly apply a slight variant of the
  algorithm in \cite{SS2017} taking $F_2$ as input: that algorithm
  already keeps the depth of recursion bounded by $\log_2(n)$, but we
  should now handle the fact that we work in the hypersurface
  $\phi_1^{-1}(\phi_1(K_1))$.  Again, all of this is under the
  assumption that one can make $F_2$ satisfy the assumptions of
  Theorem~\ref{thm:mainresult}.}

\rev{We will investigate that approach in a forthcoming paper.}

\rev{Thus, the next steps to obtain nearly optimal algorithms for
  computing roadmaps of smooth real algebraic sets, without compactness
  assumptions, are:
  \begin{itemize}
  \item to study how the constructions of generalized Lagrange systems
    introduced in~\cite{SS2017} for encoding polar varieties associated to
    linear projections can be reused in our context;
  \item to prove that assumption $(\sfB)$ holds for some generic choice of $\aa$
    and $\bb$ for our polar varieties, which by contrast to those used
    in~\cite{SS2017} are no more associated to linear projections;
  \item to prove that the variant of the algorithm designed
    in~\cite{SS2017} discussed above still has a complexity similar to
    the one obtained in~\cite{SS2017}.
  \end{itemize}
}

\rev{The example below illustrates how this whole machinery might work
  and how Theorem~\ref{thm:mainresult} can already be used.}
\begin{example}\label{exa:exroadmap}
Let $V=\V(g)\subset \CCo^3$ be the hypersurface defined by the vanishing set of 
the polynomial $g=x_1^3+x_2^3+x_3^3-x_1-x_2-x_3-1\in\QQo[x_1,x_2,x_3]$. As a 
hypersurface, $V$ is 2-equidimensional and since $\sing(V)=\emptyset$, $V$ 
satisfies $(\sfA)$.

  Let $\revbis{\bphi=((x_1-1)^2+x_2^2+x_3^2,\:x_1,\:x_2)} \subset 
\QQo[x_1,x_2,x_3]$. As the restriction of $\map{\phi}{1}$ to $\RRo^n$ is the 
square of the \revbis{Euclidean distance to $(1,0,0)$}, $(\sfP)$ is satisfied.
Since $2\leq i \leq d$, we must take $i=2$. Then we see that one can 
write
\[
    W_2=\V(f,
    \revbis{(3x_1x_3+1)(x_1-x_3)+ 3x_3^2 - 1} 
).
\]
One checks that $W_2$ is 1-equidimensional and has no singular point as well, so
that $(\bphi,2)$ satisfies $(\sfB_1)$. Let $K_2 = \Wo(\map{\phi}{1},W_2)$,
\rev{which is} a finite set of \rev{cardinality} \revbis{$45$ (of which $5$ are real)}. 
Besides, for any $\alpha\in \CCo$,
\[
	V \cap \map{\phi}{1}[-1](\alpha) = 
	\revbis{\V(f,(x_1-1)^2+x_2^2+x_3^2-\alpha)}
\]
is either empty or an equidimensional algebraic set of dimension 1. Therefore, 
$(\bphi,2)$ satisfies $(\sfB)$.
Finally, since $\Wo(\map{\phi}{1},W_2)\cap \RRo^3$ is a finite set, assumption 
$(\sfC)$ holds vacuously. \rev{Recall that, by definition, $F_2 =
  {\map{\phi}{1}}^{-1}(\map{\phi}{1}(K_2))\cap V $.} In conclusion, by
Theorem~\ref{thm:mainresult}, $W_2 \cup F_2$ is a $1$-roadmap of
$(V,\emptyset)$. Figure~\ref{fig:exroadmap} illustrates this example.

\begin{figure}[h]\centering
\begin{minipage}[c]{0.45\linewidth}
  \includegraphics[width=\linewidth]{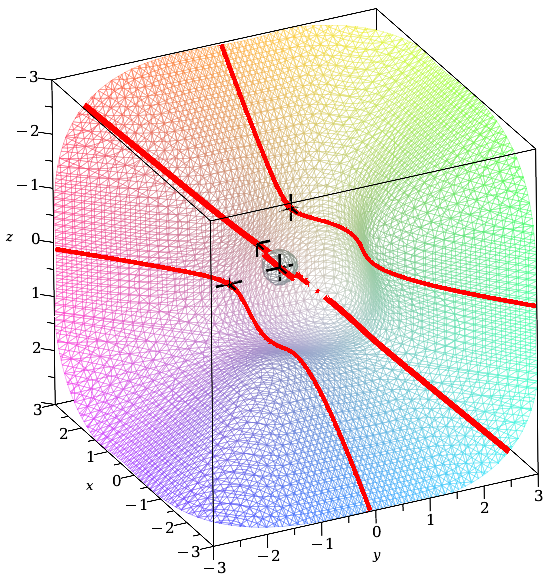}
\end{minipage}\hfil
\begin{minipage}[c]{0.45\linewidth}
  \includegraphics[width=\linewidth]{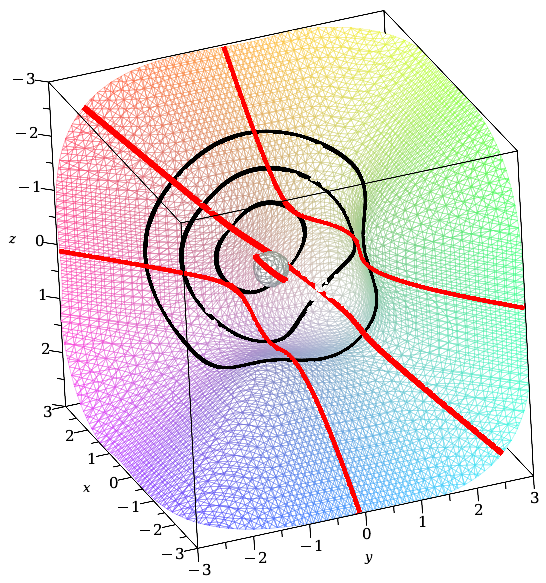}
\end{minipage}
\caption{An illustration of Example~\ref{exa:exroadmap}. The real trace $V\cap 
\RRo^3$ is plotted twice as a grid. On the left, $W_2\cap\RRo^3$ is represented 
as red lines, and the crosses represent all the real points of $K_2$. Then, on 
the right, we replaced the points of $K_2$ by the fibers of $F_2\cap\RRo^3$ 
(black lines), to repair the connectivity failures of $W_2\cap\RRo^3$. In 
particular, $F_2\cap\RRo^3$ connects the \SAC components of $W_2\cap\RRo^3$ 
that lie in the same \SACC of $V\cap\RRo^3$.} 
\label{fig:exroadmap}
\end{figure}
\end{example}

  \revbis{We expect that algorithmic progress on the computation of roadmaps for
  real algebraic and semi-algebraic sets will lead to implementations that will
  automate the analysis of kinematic singularities for e.g. serial and parallel
  manipulators. In particular, there are many families of robots where these
  algorithms could be used if they scale enough. This is the case for e.g. 6R
  manipulators (see e.g. the results on the number of aspects in \cite{Wenger07}
  which need to be extended) in the context of serial manipulators, for the 
  study of self-motion spaces of parallel platforms such as Gough-Stewart ones 
  (the case of such manipulators with $6$ lengths still remains open, see e.g.
  \cite{NaSc17}) and for the identification of cuspidality manipulators (see 
  \cite{CPSSW2022} for a general approach, relying on roadmap algorithms).
  For some of these applications, one needs to compute the
  number of connected components of semi-algebraic sets defined as the
  complement of a real hypersurface defined by $f = 0$ where $f$ is a
  multivariate polynomial. Note that this can be done by computing a roadmap
  for the (non-bounded) real algebraic set defined by $tf-1=0$ where $t$ is a
  new variable. This illustrates the potential interest of the algorithms that
  would be derived from the connectivity theorem of this paper.}

\section*{Acknowledgments}
The first and second authors are supported by ANR grants ANR-18-CE33-0011 
\textsc{Sesame}, and ANR-19-CE40-0018 \textsc{De Rerum Natura}, the joint 
ANR-FWF ANR-19-CE48-0015 \textsc{ECARP} project, the European Union’s Horizon 
2020 research and innovative training network program under the Marie 
Skłodowska-Curie grant agreement N\textsuperscript{o} 813211 (POEMA) and the 
Grant FA8665-20-1-7029 of the EOARD-AFOSR. The third author was supported by an 
NSERC Discovery Grant.
\rev{We also thank the referees for their helpful remarks and suggestions.}
\bibliographystyle{abbrv}
\bibliography{biblio.bib}

\end{document}